\documentclass[11pt]{article}

\usepackage{times,amsmath,amssymb,amsthm,latexsym,fullpage}
\usepackage{algorithm}
\usepackage[noend]{algorithmic}
\algsetup{linenosize=\scriptsize}

\usepackage{paralist}
\usepackage{multicol}
\usepackage{cite}
\usepackage{hyperref}
\usepackage{comment}

\usepackage{color}
\usepackage{tikz}
\usetikzlibrary{arrows}
\usetikzlibrary{positioning}

\newcommand{\faith}[1]{{\color{red} #1}\normalcolor}
\newcommand{\sapta}[1]{{\color{blue} #1}\normalcolor}

\def\nats{{\mathbb N}}

\newtheorem{theorem}{Theorem}

\newtheorem{lemma}[theorem]{Lemma}

\newtheorem{observation}[theorem]{Observation}

\renewcommand{\paragraph}[1]{\medskip\noindent{\bf #1}:\ }

\newcommand{\remove}[1]{}
\newcommand{\ExAbsShrink}[1]{}

\newcommand{\AlgName}{{\sc CCReg}}
\newcommand{\knows}[3]{$\mbox{\it SysInfo}^{{#3}} \subseteq \mbox{\it Changes}_{{#1}}^{{#2}}$}

\begin{document}

\title{Simulating a Shared Register in a System that Never Stops Changing\thanks{
       A preliminary version of this paper appeared in DISC 2015~\cite{AttiyaCEKW2015}.
       This work is supported by the Israel Science Foundation (grants 1227/10 and 1749/14),
       by the Natural Science and Engineering Research Council of Canada,
       and by the US National Science Foundation grants 0964696 and 1526725.}}

\author{Hagit Attiya\footnotemark[1]
        \and Hyun Chul Chung\footnotemark[2]\thanks{
            Hyun Chul Chung is currently working at Epoch Labs, Inc. Austin, TX USA.}
        \and Faith Ellen\footnotemark[3] \\
        \and Saptaparni Kumar\footnotemark[2]
        \and Jennifer L.\ Welch\footnotemark[2]}

\maketitle

\renewcommand{\thefootnote}{\fnsymbol{footnote}}
\footnotetext[1]{Department of Computer Science, Technion}
\footnotetext[2]{Department of Computer Science and Engineering, Texas A\&M University}
\footnotetext[3]{Department of Computer Science, University of Toronto}
\renewcommand{\thefootnote}{\arabic{footnote}}

\begin{abstract}
Simulating a shared register can mask the intricacies of designing algorithms
for asynchronous message-passing systems subject to crash failures, since it allows them to run algorithms
designed for the simpler shared-memory model.
Typically such simulations replicate the value of the register in
multiple servers
and require
readers and writers
to communicate with a majority of servers.
The success of this approach for static systems,
where the set of
nodes (readers, writers, and servers)
is fixed,
has motivated several
similar simulations for dynamic systems,
where nodes may enter and leave.
However, existing simulations
need to assume that the system
eventually stops
changing
for a long enough period
or that the system size is bounded.

This paper presents the first simulation of an atomic 
read/write register in a crash-prone asynchronous system that can change size and 
withstand
nodes continually entering and leaving.
The simulation allows the system to keep changing,
provided that the
number of nodes entering and leaving during a fixed time interval
is at most a constant fraction of the current system size. The simulation 
also 
 tolerates node crashes
as long as the number of failed nodes in the system is at most 
a constant fraction of the current system size. 
\end{abstract} 

\section{Introduction}

Simulating a shared read/write register is a way to mask the intricacies of
designing algorithms for asynchronous message-passing systems subject
to crash failures, since it allows them to run algorithms designed for
the simpler shared-memory model.
Typically, such simulations replicate the value of the register 
in multiple servers and require
readers and writers to communicate with a majority of servers.

Most of the work in this area has focused on simulating  atomic shared registers.
 For example, the ABD simulation~\cite{AttiyaBD1995} replicates the value of the register in
server nodes. It assumes that a majority of the server nodes do not fail.
Consider the simplified case of a single writer and a single reader.
To write the value $v$, the writer sends $v$, tagged with a sequence number,
to all servers and waits for acknowledgements from a majority of them.
Similarly, to read, the reader contacts all servers,
waits to receive values from a majority of them, and then
returns the
value
with the highest sequence number.
This approach can be extended to the case of multiple writers and multiple
readers by having each operation consist of a read phase, used by a writer
to determine its sequence number and used by a reader to obtain the return
value, followed by a write phase, used by a writer to disseminate the value
(and sequence number) and used by a reader to announce the sequence number
of the value it is about to return~\cite{LynchS1997}.

The success of this approach for static systems, where the set of
readers, writers, and servers
is fixed, has motivated several
similar simulations for {\em dynamic} systems, where nodes may enter
and leave.
Change in system composition due to nodes entering and leaving is called {\em churn}.
However, existing simulations of atomic  
registers rely either on the assumption that churn eventually stops for a long enough period
(e.g., DynaStore~\cite{AguileraKMS2011} and RAMBO~\cite{LynchS2002})
or on the assumption that the system size is bounded (e.g.,~\cite{BaldoniBKR2009}). 
See Section~\ref{section:related} for a detailed discussion of related work.

In this paper, we take a different approach: {\em we allow
churn to continue forever, while still ensuring that read
and write operations complete and nodes can join and leave the system.}
Our churn model puts an upper bound on the number of nodes
that can enter or leave during any time interval of a certain length.
The upper bound is a constant 
fraction of the number of nodes that are present in the system at the
beginning of the time interval.
So, as the system size grows, the allowable number of changes to its composition
grows as well. Similarly, as the system size shrinks, the allowable
number of changes shrinks.


\remove{\faith{We assume an asynchronous model in which there is
an \emph{unknown} upper bound $D$ on the delay
of any message (between nonfaulty nodes).
The delay of a message can be an arbitrarily small positive value.}}

In more detail, our churn model relies on a parameter $D$ of the system model, which is an upper bound, unknown to the nodes,
 on the delay of any message (between nodes that have not crashed). It is important to note that we set no lower bound on the delay of
messages, so consensus cannot be solved in this model,
even in the static case with no nodes entering or leaving 
 and the possibility of one node crashing. We assume that, in any time interval of length $D$,
the number of nodes that can enter or leave the system 
is at most a constant 
fraction, $\alpha$, of the number of nodes in the system at the beginning of the interval. 
The constant $\alpha$ is known to all  nodes. For example, if the churn rate is $\alpha = 0.05$ and, at time $t$, the system contains $100$ nodes, at most $5$ nodes can enter or leave the system within the interval $[t, t+D]$.

We believe ours is an appealing and a reasonable churn model.
For instance, if each node has the same probability of leaving in a time interval,
then the number of leaves is expected to be a fixed fraction of the total number
of nodes.
(See \cite{KoHG2008} for a discussion of churn behavior in practice.)

Our algorithm tolerates crash failures of nodes, as well as churn.
In the preliminary version  of this paper~\cite{AttiyaCEKW2015}, we assumed that the
number of nodes 
that crash is bounded above by a fixed constant, $f$,
independent of the system size.
Here, we only require that the number of crashed nodes is always at most a constant fraction, $\Delta$, of the system size. The constant $\Delta$ is known to all nodes. 
As a consequence, if the number of crashed nodes at some point in an execution
is $\Delta$ times the system size, then no 
nodes can leave. To make the model more dynamic,
we allow, but do not require, the adversary (who is responsible for crashing nodes)
to notify a node, $p$, which has not crashed, that a crashed node, $q$, has left.
Then $p$ sends a leave message on behalf of node $q$ 
(analogously to~\cite{AguileraKMS2011} and~\cite{LynchS2002}).  
Such {\em forced leaves} are  counted as part of the churn.

Our algorithm, called \AlgName{} (for \emph{Continuous Churn Register}),
is intuitive, combining the simple static algorithm for multiple readers and multiple
writers outlined above with a joining protocol
and 
careful
estimations of the number of nodes from which
responses should be received for joining, reading, and writing.
In order to join, a newly entered node announces its entry,
waits to receive sufficiently many acknowledgements,
and then 
announces it has joined. Once a node has joined, it can perform reads and writes.
A node leaves the system by announcing its departure.
Each node maintains a set of changes to the composition of the system,
based on the announcements of nodes entering, joining and leaving.
This information is also propagated through appropriate echo messages
and by having each node append the set of changes it has seen
to its messages that echo enter announcements.

When a node first receives an acknowledgement of its entrance announcement
from a node that has already joined, it calculates the number of acknowledgements it needs
to join as a fraction (depending on $\alpha$ and $\Delta$, which it knows) of the number of nodes it
believes are in the system.
To ensure that information about the system composition is propagated properly,
it is crucial that before a node joins, it gets at least one acknowledgement from a node whose
information is up to date. This is ensured by requiring that that the number of acknowledgements
it receives before joining is sufficiently large, so that at least one of them is from a node that has been in the system for sufficiently long. 
\remove{A joining node
calculates the number of acknowledgements it needs as
a fraction (depending on $\alpha$ and $\Delta$) of the number of nodes it
believes are in the system when it first receives an acknowledgement
from a node that has already joined.
It is crucial that this number be large enough
to ensure that at least one acknowledgement is from a node $p$ that has been in
the system sufficiently long,
so that $p$ has up-to-date information.
This ensures that information about the system composition is propagated properly.}
The number of necessary acknowledgements must also be small enough
to ensure that the node will eventually receive enough
acknowledgements.
\remove{\faith{IN THE DESCRIPTION OF THE MODEL AND CCREG ABOVE, WE ONLY TALK ABOUT nodes,
NOT servers. SUDDENLY, IN THE NEXT PARAGRAPH, WE STARTED TALKING ABOUT servers.
I HAVE CHANGED THE 3 OCCURRENCES OF "servers" IN THE FOLLOWING PARAGRAPH
TO "nodes". IS THIS OKAY?}}

Each reader and writer keeps
track of the number of nodes that have joined,
but not left. We call these \emph{members}.
The read and write phases of operations
wait for responses from a constant fraction of
the nodes believed to be members.
As in the joining protocol, the number of responses must be
small enough so that termination is guaranteed.
To prove \AlgName{} is atomic, 
 we consider two cases:
If a read occurs shortly after a write, then we must ensure that
the sets of nodes contacted by the two operations are intersecting.
This is analogous to the situation in the static, majority algorithm.
If operations are farther apart in time, then,
as in the join protocol, we ensure that information about writes
to the register is propagated properly.

 Our churn model has the pleasing property that it is 
algorithm-independent:  It only refers to nodes that enter or leave
irrespective of whether an entered node completes the join protocol.

\section{Model}
\label{section:model}

We consider an asynchronous message-passing system, with nodes running
\emph{client} threads (\emph{reader} threads or \emph{writer} threads)
and \emph{server}
threads.  Nodes do not have clocks, so they cannot determine the
current time, nor directly measure how much time has elapsed since some
event.  Each node runs exactly one server thread, at most one reader
thread, and at most one writer thread.

Nodes can enter and leave the system during an execution.  We model
this behavior by assuming the existence of an adversary that generates
{\sc Enter}$(p)$ and {\sc Leave}$(p)$ signals to indicate that $p$
should enter or leave; only $p$ experiences these signals (with one possible
exception explained below).  For each node $p$, there is at most
one {\sc Enter}$(p)$ signal and at most one {\sc Leave}$(p)$ signal.
Thus a node that leaves the system cannot re-enter the system.  (This
restriction is easy to remove by giving a new name to a node that
wants to re-enter.)

We say that a node is {\em present} at time $t$ if it has entered the
system (i.e., an {\sc Enter}$(p)$ signal has occurred) but has not
left by time $t$ (i.e., no {\sc Leave}$(p)$ signal has occurred so
far).  We let $N(t)$ denote the number of servers whose nodes are
present at time $t$; $N(t)$ is called the \emph{system size}.  We
assume that there are always at least $N_{min}$ servers whose nodes
are present in the system, i.e., at all times $t$, $N(t) \geq
N_{min}$.  Let $S_0$ denote the set of nodes that are present
initially, i.e.,~at time $0$.  Note that $|S_0| = N(0)$.

Nodes are subject to crash failures.  The adversary generates a
{\sc Crash}$(p)$ signal at node $p$ to indicate that $p$ has crashed.
There can be at most one {\sc Crash}$(p)$ signal for each node $p$.  A
crashed node does not take any more steps, does not send any more
messages, and no more messages are delivered to it.
We say that a node is {\em active} at time $t$ if it is present at time $t$
and has not crashed by time $t$. 

Nodes communicate through a broadcast service that provides a
mechanism to send the same message to all nodes in the system.
If a
server wants to send a message to a single client, it can do so by
broadcasting the message and indicating that the message should be
ignored by the other clients.
Message delays are bounded above by a
system parameter, $D$, that is unknown to the nodes.
In more detail, a
message that is broadcast by a node $p$ at time $t$ is guaranteed to
arrive at each node $q \neq p$ within $D$ units of time, provided that
$q$ is active throughout the interval $[t, t+D]$. 
If $q$ is active for
some, but not all, of $[t,t+D]$, then $q$ might or might not receive the
message.
Nodes that enter after time $t+D$ do not (directly) receive
the message.
A message can take a different amount of time to reach different nodes. 
All messages broadcast by $p$ are received by $q$ in the same
order in which $p$ sent them.  In addition to the maximum transmission
delay, $D$ includes the maximum time for handling the message at both
the sender and the receiver.  There is no lower bound on the actual
length of time it takes for a message to be transmitted, nor on the
amount of time to perform local computation at a node (i.e.,~they
could take an arbitrarily small amount of time). 
An execution in which all messages satisfy these constraints is called {\em valid}. 

Since there is no bound on the ratio between the fastest and slowest
messages, the system is essentially asynchronous and consensus cannot
be solved in our system model.  In fact, any problem that can be
solved in our model, even if the bound $D$ is known by all nodes, can
be solved in the standard asynchronous message passing model, where
there is no upper bound on message delivery time.  To see why,
consider an algorithm $\mathbb{A}$ that is designed to work when there
is a known upper bound $D$ on message delivery time.  Now consider any
execution $e$ of this algorithm in the standard asynchronous message
passing model. We compress it into an execution $e'$ by changing the
real time of the occurrence of the $i^{th}$ event to be $1 -
2^{-i}$. Then every message that is received by a node in $e$ is
received within time $D=1$ of when it was sent in $e'$.  Moreover, if,
in the original execution, messages are received along a link in the
order they were sent, then the same is true in this timed execution.
Since $e'$ is a valid execution of $\mathbb{A}$ in our model with
$D=1$, $e$ is a 
valid 
execution of $\mathbb{A}$, with each node
having the same local execution history in both executions.

We assume the set of nodes that are present does not change too
quickly: There is a constant $\alpha < 1$, known to all nodes, called
the \emph{churn rate}, such that for all times $t$, at most $\alpha
\cdot N(t)$ nodes enter or leave (i.e, experience {\sc Enter} or {\sc
  Leave} signals) during the interval $[t,t+D]$.  This is a constraint
on the adversary.

There is also a constant $\Delta < 1$, known to all nodes, called the
\emph{failure fraction}, such that, at any time $t$, at most $\Delta
\cdot N(t)$ of the nodes present at time $t$ have crashed (i.e., have
experienced a {\sc Crash} signal previously).  This is another
constraint on the adversary.
If $\Delta \cdot N(t)$ nodes are crashed at time $t$, then
active nodes cannot leave the system. 

At any time, 
the adversary may choose to let some of the crashed nodes leave the
system.
Since a crashed node $p$ cannot take any actions on its own behalf,
this is accomplished by the adversary generating a {\sc Leave}$(p)$
signal at an active node $q$.
We call these {\em forced leaves} and they
contribute to the churn quota.  That is, for any time $t$, the sum of
the number of {\sc Enter} and {\sc Leave} signals (including a
{\sc Leave}$(p)$ 
signal 
that occurs at a node other than $p$) that
happen in $[t,t+D]$ is at most $\alpha \cdot N(t)$.  This mechanism works even when the adversary generates a {\sc Leave}$(p)$ signal at more than one active nodes. However all {\sc Leave}$(p)$ signals for a crashed node $p$ contribute to only one forced leave.
It is important
to note that forcing crashed nodes to leave is entirely optional
for the adversary. 
However, letting the adversary have this ability allows
active nodes to leave the system or 
more crashes to be accommodated
if $\Delta \cdot N(t)$ nodes are crashed at time $t$, 
even if no more nodes
enter 
after time $t$
.\footnote{This mechanism could reflect the output of some
failure detection method that is outside the scope of our paper.} 
\remove{\faith{There is one difficulty with accommodating failure detection in our model:
What if more than one node detects that node $p$ has crashed and generates
LEAVE($p$) signals? Does our algorithm still work if the adversary can generate
a LEAVE($p$) signal at more than one active node for the same crashed node $p$?
If so, we should say so. SAPTA: Discuss}}

In this model, we want to simulate a multi-reader multi-writer atomic
read-write register.  A user application at each node determines when
read and write operations are invoked at the node, subject to two
constraints.  The first constraint is that no read or write is invoked
on a node unless the node is active and it is ready to accept
operations. 
The node indicates that it is ready by 
broadcasting a {\em joined} message. 
The second constraint is that the user does not invoke
a read or write operation at a node if there is a previous read or
write operation that it has invoked but is not yet completed.

\remove{\faith{WHAT DOES IT MEAN THAT A NODE EXECUTES A JOIN EVENT? WHY IS THIS PART OF THE MODEL?
  IN THE PRECEDING PARAGRAPH, I CHANGED executing a {\em join} event." TO "broadcasting a {\em joined} message." OR WOULD IT BE BETTER TO
  JUST REMOVE THE SENTENCE " The node indicates that it is ready by executing a {\em join} event." ?}}

To sum up, the adversary determines when nodes
enter, leave, and crash, 
the network 
(which is controlled by the adversary) 
determines when messages are delivered,
and the users 
(which are also controlled by the adversary) 
determine when reads and writes are invoked, subject to
the constraints discussed above.  The algorithm running at the nodes
is responsible for joining and generating responses to read and write
invocations.

We consider an algorithm to be {\em correct} if every execution of the
algorithm in the model just described satisfies the following conditions:
\begin{itemize}
\item Every active process that does not leave or crash eventually joins.
\item Every read or write that is invoked at a process that does not
      leave or crash eventually completes.
\item The read and write operations satisfy atomicity:  
 there is an ordering of all completed reads and writes and some subset of the
      uncompleted writes such that
      every read returns the value of the latest preceding write and, if 
      an operation
      $op_1$ finishes before another operation $op_2$ begins, then $op_1$
      is ordered before $op_2$.
\end{itemize}

\makeatletter
\newcommand{\setalglineno}[1]{%
  \setcounter{ALC@line}{\numexpr#1-1}}
\makeatother

\section{The \AlgName{} Algorithm}

The algorithm combines a mechanism for tracking the composition of
the system, with a simple algorithm, very similar
to~\cite{LynchS1997},
for reading and writing the register,
which associates a unique timestamp with each value that is written. 

\remove{\faith{NOTE THAT THE TIMESTAMP OF A REGISTER VALUE IS USED 2 PARAGRAPHS LATER,
SO I THOUGHT IT MIGHT BE GOOD TO INTRODUCE IT HERE.}}

In order to track the composition of the system
(Algorithm~\ref{algo:Common}),
each node $p$ maintains a set of events, $Changes_p$,
concerning the nodes that have entered the system.
When a node $q$ enters, it adds $enter(q)$ to $Changes_q$
and broadcasts an enter message requesting information about prior events.
We say that $q$ enters or the $enter(q)$ event occurs at the time, $t_q^e$, when this broadcast is sent. 
When a node $p$ finds out that $q$ has entered the system,
either by receiving this message or by learning indirectly from another node,
it adds $enter(q)$ to $Changes_p$.
When $q$ has received sufficiently many messages in response to its request,
it knows relatively accurate information about prior events and the value of
the register.
(Setting the \emph{join bound}, denoted by $\gamma$,
on the number of messages that should be received
is a key challenge in the algorithm.)
When this happens, $q$ adds $join(q)$ to $Changes_q$,
sets its $is\_joined_q$ flag to $true$,
and broadcasts a message saying that it has joined.
We say that $q$ {\em joins} 
or the $join(q)$ event occurs at the time, $t_q^j$, 
when this broadcast is sent.
When $p$ finds out that $q$ has joined,
either by receiving this message or by learning indirectly from another node,
it adds $join(q)$ to $Changes_p$.
When $q$ leaves, it simply broadcasts a leave message
or another node broadcasts a leave message on $q$'s behalf. 
We say that the $leave(q)$ event occurs at the time, $t_q^\ell$, when this broadcast is sent. 
When $p$ finds out that $q$ has left the system,
either by receiving this message or by learning indirectly from another node,
it adds $leave(q)$ to $Changes_p$.

When a node $p$ receives an enter message from a node $q$,
it responds with an enter-echo message containing $Changes_p$,
its current estimate of the register value (together with its timestamp),
$is\_joined_p$ (indicating whether $p$ has joined yet), and $q$.
When $q$ receives an enter-echo in response (i.e., that ends with $q$),
it increments its {\em join-counter}.
The first time $q$ receives such an enter-echo from a joined node,
it computes $join\_bound$,
the number of enter-echo messages it needs in response before it can join.

\remove{\faith{I CHANGED THE ORDER OF LINES 1 AND 2 IN THE CODE FOR ALGORITHM 1, TO BE CONSISTENT WITH THE OTHER CASES AND THE DESCRIPTION IN THE TEXT. I HOPE THIS DOESN'T REQUIRE ANYTHING TO BE CHANGED IN THE PROOF. I ALSO REMOVED "combined with next variable to make a unique timestamp for the write;" BECAUSE THIS IS
NOW EXPLAINED IN THE ACCOMPANYING TEXT.}}

\begin{algorithm*}[tb]
\begin{algorithmic}[1]
	\small
	\item[] {\bf Local Variables:}
	\item[] $is\_joined$ \COMMENT{Boolean to check if $p$ has joined the system; initially $false$}
	\item[] $join\_counter$ \COMMENT{for counting the number of enter-echo messages received by $p$; initially $0$}
	\item[] $join\_bound$ \COMMENT{if non-zero, the number of enter-echo $p$ should receive before joining; initially $0$}
	\item[] $\mbox{\it Changes}$ \COMMENT{set of 
	$enter$, $leave$,  and $join$ events} known by $p$; 
	\item[] \hspace*{.5in}\COMMENT{initially 
	$\{enter(q) ~|~ q \in S_0\} \cup \{ join(q) ~|~ q \in S_0\}$, 
	 if $p \in S_0$, and $\emptyset$, otherwise}
        \item[] $val$ \COMMENT{latest register value known to $p$; initially $\perp$}
        \item[] $seq$ \COMMENT{sequence number of latest value known to $p$; initially 0}
        \item[] $id$ \COMMENT{id of node that wrote latest value known to $p$;
                initially $\perp$}
	\item[]
	\item[] {\bf Derived Variable:}
	\item[] $\mbox{\it Present}	 = \{q ~|~ enter(q) \in \mbox{\it Changes} \wedge leave(q) \not\in \mbox{\it Changes}  \}$
	\item[] \hrulefill	
\begin{multicols*}{2}
	\item[] {\bf When $p$ receives {\sc Enter}(p) signal:}
	\STATE add $enter(p)$ to $\mbox{\it Changes}$
	\STATE bcast $\langle$``enter'', $p\rangle$

	\item[]	
	
	\item[] {\bf When $\langle$``enter'', $q\rangle$ is received:}
        \STATE add $enter(q)$ to $\mbox{\it Changes}$
	\STATE bcast $\langle$``enter-echo'', $\mbox{\it Changes}$,\\
            ~~~$(val,seq,id)$, $is\_joined$, $q\rangle$	
	\item[]	
	
	\item[] {\bf When $\langle$``enter-echo'', $C$, $(v,s,i)$, $j$, $q\rangle$}\\
            ~~~{\bf is received:}
	\IF{$(s,i) > (seq,id)$}            \label{line:new value in enter?}
		\STATE $(val,seq,id) := (v,s,i)$ \label{line:update value in enter}
	\ENDIF
	\STATE $\mbox{\it Changes} := \mbox{\it Changes} \cup C$

	\IF{$\neg is\_joined \wedge (p = q)$}
		\IF{$(j = true)\!\wedge\!(join\_bound = 0)$}
			\STATE 	$join\_bound :=  \gamma \cdot | \mbox{\it Present} | $\label{line:calculate join bound}
		\ENDIF
		\STATE $join\_counter$++
		\IF{$join\_counter\!\geq\! join\_bound\! >\! 0$} \label{line:check if enough enter echoes}
			\STATE $is\_joined := true$
			\STATE add $join(p)$ to $\mbox{\it Changes}$
			\STATE bcast $\langle$``joined'', $p\rangle$
		\ENDIF
	\ENDIF

	\columnbreak
	
	\item[] {\bf When $\langle$``joined'', $q\rangle$ is received:}	
	\STATE add $join(q)$ to $\mbox{\it Changes}$
	\STATE add $enter(q)$ to $\mbox{\it Changes}$
	\STATE bcast $\langle$``joined-echo'', $q\rangle$
	\item[]	
		
	\item[] {\bf When $\langle$``joined-echo'', $q\rangle$ is received:}	
	\STATE add $join(q)$ to $\mbox{\it Changes}$
	\STATE add $enter(q)$ to $\mbox{\it Changes}$
	\item[]	
			
	\item[] {\bf When $p$ receives {\sc Leave}(q) signal:}
	\STATE bcast $\langle$``leave'', $q\rangle$	
	\IF{$p=q$}
	\STATE halt
	\ENDIF
	\item[]	
	
	\item[] {\bf When $\langle$``leave'', $q\rangle$ is received:}
	\STATE add $leave(q)$ to $\mbox{\it Changes}$
	\STATE bcast $\langle$``leave-echo'', $q\rangle$	
	\item[]
	
	\item[] {\bf When $\langle$``leave-echo'', $q\rangle$ is received:}	
	\STATE add $leave(q)$ to $\mbox{\it Changes}$	
	
\end{multicols*}
\end{algorithmic}
\caption{\AlgName---Common code managing the \mbox{\it Changes} variable, for node $p$.}
\label{algo:Common}
\end{algorithm*}

Once a node has joined, its reader and writer threads can handle
read and write operations.
A node is a \emph{member} at time $t$
if it has joined, but not left, by time $t$.
Initially, $Changes_p = \{enter(q)~|~ q \in S_0\} \cup  \{join(q) ~|~ q \in S_0\}$,
if $p \in S_0$, and $\emptyset$ otherwise.
A node $p$ also maintains the set
$Present_p = \{q ~|~ enter(q) \in Changes_p \wedge leave(q) \not\in Changes_p\}$
of nodes that $p$ thinks are present,
i.e.,~nodes that have entered, but have not left, as far as $p$ knows.
\remove{\faith{Present AND Changes ARE ONLY USED BY THE NODE (Algorithm 1), NOT BY THE READER AND WRITER THREADS
(Algorithm 2) OR THE SERVER THREAD (Algorithm 3).}}
The client at node $p$ maintains the derived variable $Members_p = \{q ~|~ join(q) \in Changes_p \wedge leave(q) \not\in Changes_p\}$
of nodes that $p$ thinks are members.

The client thread treats read and write operations in a similar manner
(Algorithm~\ref{algo:rw_client}). We assume that the code segment that is executed in response to each event executes without interruption.  
Both operations start with a read phase,
which requests 
the current value of the register, using a query message,
followed by a write phase, using an update message.
A write operation broadcasts the new value it wishes to write,
together with a timestamp, which consists of a sequence number
that is one larger than the largest sequence number it has seen
and its id that used to break ties.
A read operation just broadcasts the value it is about to return,
keeping its sequence number. As in~\cite{AttiyaBD1995},
write-back is needed to ensure the atomicity of read operations. 
Both the read phase and the write phase wait to receive
sufficiently many response messages.
(Again, setting the \emph{quorum bound}, denoted $\beta$,
on the number of messages that should be received
is a key challenge in the algorithm.)


A client $p$ maintains a sequence number,
$tag$,
which it increments at the beginning of each read phase.
This is used to identify responses
belonging to its current 
read or write phase.

The server thread is simple (Algorithm~\ref{algo:Server}).
The nodes uses the variables {\em val}, {\em seq}, and {\em id} to store the latest value of the register it knows about (in {\em val}) and that value's associated timestamp (in {\em seq} and {\em id}).
When 
the server
receives an update message with a 
larger timestamp, 
it updates the 
value 
and the timestamp. 
(Note that 
timestamps, which consist of $(seq,id)$ pairs, are ordered 
lexicographically.)
When a server receives a query, it responds with the
value and its timestamp.

\begin{algorithm*}[tb]
\begin{algorithmic}[1]
\setalglineno{30}
	\small
	\item[] {\bf Local Variables:}
	\item[] $temp$ \COMMENT{temporary storage for the value being read or written; initially $0$}
	\item[] $tag$ \COMMENT{used to uniquely identify read and write phases of an operation; initially $0$}
	\item[] $quorum\_size$ \COMMENT{stores the quorum size for a read or write phase; initially $0$}
	\item[] $heard\_from$ \COMMENT{the number of responses/acks received for a read/write phase; initially $0$}
	\item[] $rp\_pending$ \COMMENT{Boolean indicating whether a read phase
             is in progress; initially $false$}	
	\item[] $wp\_pending$ \COMMENT{Boolean indicating whether a write phase
             is in progress; initially $false$}	
	\item[] $read\_pending$ \COMMENT{Boolean indicating whether a read is in progress; initially $false$}
	\item[] $write\_pending$ \COMMENT{Boolean indicating whether a write is in progress; initially $false$}
	\item[] {\bf Derived Variable:}
	\item[] 
	$\mbox{\it Members} = \{q ~|~ join(q) \in \mbox{\it Changes} \wedge leave(q) \not\in \mbox{\it Changes}\}$
	\item[] \hrulefill	
\begin{multicols*}{2}
	\item[] {\bf When READ is invoked:}
	\STATE $read\_pending := true$
	\STATE call BeginReadPhase()
	\item[]

	\item[] {\bf When WRITE($v$) is invoked:}
	\STATE $write\_pending := true$
	\STATE $temp := v$
	\STATE call BeginReadPhase()
	\item[]
		
	\item[] {\bf Procedure} BeginReadPhase()	
	\STATE $tag$++
	\STATE bcast $\langle$``query'', $tag$, $p\rangle$  \label{line:bcast query}
	\STATE $quorum\_size :=  \beta |\mbox{\it Members}|$ 
	\STATE $heard\_from := 0$
	\STATE $rp\_pending := true$	
	\item[]	
	
	\item[] {\bf When $\langle$``response'', $(v,s,i), rt
	, q
	\rangle$}\\
            ~~~{\bf  is received:}
	\IF{$rp\_pending  \wedge (rt = tag) 
	\wedge (q=p)
	$}
		\IF{$(s,i) > (seq,id)$}
			\STATE $(val,seq,id) := (v,s,i)$
		\ENDIF
		\STATE $heard\_from$++                \label{line:inc heard from resp}
		\IF{$heard\_from \ge quorum\_size$}   \label{line:quorum reached}
			\STATE $rp\_pending := false$
			\STATE call BeginWritePhase()
		\ENDIF
	\ENDIF
	
	\columnbreak	
	
	\item[] {\bf Procedure} BeginWritePhase()
    \IF{$write\_pending$}		
		\STATE  $val := temp$  \label {line:new timestamp1} 
		\STATE $seq$++\label {line:new timestamp}  

		\STATE $id := p$ \label {line:new timestamp2}  
	    
	\ENDIF
	\IF{$read\_pending$}
	\STATE $temp:= val$
	\ENDIF
	\STATE bcast $\langle$``update'', $(temp,seq,id)$,$tag, p \rangle$ \label{line:bcast update1} \
	\STATE $quorum\_size :=  \beta |\mbox{\it Members}| $
	\STATE $heard\_from := 0$
	\STATE $wp\_pending := true$
	\item[]
	
	\item[] {\bf When $\langle$``ack'', $wt
	, q
	 \rangle$ is received:}
	\IF{$wp\_pending \wedge (wt = tag)
	\wedge (q=p)
	$}
		\STATE $heard\_from$++                \label{line:inc heard from ack}
		\IF{$heard\_from \ge quorum\_size$}
			\STATE $wp\_pending := false$
			\IF{$read\_pending$}
				\STATE $read\_pending := false$
				\STATE RETURN $temp$
			\ENDIF			
			\IF{$write\_pending$}
				\STATE $write\_pending := false$
				\STATE ACK			
			\ENDIF
		\ENDIF
	\ENDIF
	
\end{multicols*}
\end{algorithmic}
\caption{\AlgName---Client code, for node $p$.}
\label{algo:rw_client}
\end{algorithm*}

\remove{\faith{I CHANGED send TO bcast AND ADDED IN THE PROCESS TO WHOM THE MESSAGE IS SENT.
THIS IS CONSISTENT WITH THE SECOND SENTENCE OF PARAGRAPH 5 IN SECTION 2.}}

\begin{algorithm*}[tb]
\begin{algorithmic}[1]
\setalglineno{70}
\begin{multicols*}{2}
	\item[] {\bf When $\langle$``update'', $(v,s,i), wt, q \rangle$}\\
            ~~~{\bf is received:}
	\IF{$(s,i) > (seq,id)$}        \label{line:new value?}
		\STATE $(val,seq,id) := (v,s,i)$ \label{line:update value}
	\ENDIF
	\IF{$is\_joined$}
	        \STATE 
	        bcast
	         $\langle$``ack'', $wt
	         , q
	         \rangle$
	\ENDIF
	\STATE bcast $\langle$``update-echo'', $(val,seq,id)\rangle$

	\columnbreak

	\item[] {\bf When $\langle$``query'', $rt$, $q\rangle$ is received:}
	\IF{$is\_joined$}
		\STATE  
		bcast
		 $\langle$``response'', $(val,seq,id), rt
		 , q
		  \rangle$
	\ENDIF
	\item[]
	
	\item[] {\bf When $\langle$``update-echo'', $(v,s,i)\rangle$}\\
            ~~~{\bf is received:}	
	\IF{$(s,i) > (seq,id)$}
		\STATE $(val,seq,id) := (v,s,i)$
	\ENDIF
	
\end{multicols*}
\end{algorithmic}
\caption{\AlgName---Server code, for node $p$.}
\label{algo:Server}
\end{algorithm*}

The correctness of \AlgName{} relies on the
following assumption about the churn rate
\begin{align}
\alpha  &\leq 1 - 2^{-1/4} \approx 0.159 \label{parameters:G} 
\end{align}
and the following relation between the system parameters 
$\alpha$, $\Delta$ and $N_{min}$:
\begin{align}
1 & < \left((1- \alpha)^3 -\Delta(1+\alpha)^3\right)N_{min} \label{parameters:D}                  												\end{align}
%
Note that Assumption~(\ref{parameters:G}) restricts the churn rate $\alpha$ during an interval of length $D$
to less than 16\%. We believe this is reasonable, since $D$, the maximum message delay, will typically be quite small. 


The algorithm parameters $\gamma$ and $\beta$ must satisfy the following 
constraints:
\begin{align}															
\gamma & \geq   \frac{1}{N_{min} (1-\alpha)^3}+(1+\Delta) \frac{(1+\alpha)^3}{(1-\alpha)^3}-1                          		
															\label{parameters:H}\displaybreak[0]\\
\gamma &\leq  \frac{(1-\alpha)^3}{(1+\alpha)^3}-\Delta                \label{parameters:B}\displaybreak[0]\\
 \beta & \leq \frac{(1-\alpha)^3}{(1+\alpha)^2} - \Delta (1+\alpha)
                                                      \label{parameters:C}\displaybreak[0]\\
 \beta &>\frac{(1+\alpha)^5 - 1}{(1 - \alpha)^4}      \label{parameters:E}\displaybreak[0]\\
\beta &>\frac{ (1+\Delta)(1+\alpha)^3 - (1-\alpha)^3 + 1}{(2-2\alpha +\alpha^2)(1-\alpha)^2(1+\alpha)^{-2} }\label{parameters:F}
\end{align}
Table~\ref{table:assumptions} gives a few sets of values for which 
the above assumptions are satisfied. 
\begin{table}[bt]
\begin{center}
\begin{tabular}{ |c|c|c|c|c| }
\hline
\multicolumn{3}{|c|}{system parameters} &\multicolumn{2}{c|}{derived algorithm parameters} \\\hline
churn rate ($ \alpha$) & failure fraction ($\Delta$) & minimum system size ($N_{min}$)
& join bound ($\gamma$) & quorum size ($\beta$) \\\hline
$0$ & $0.33$ & N/A & N/A& $0.665$ \\ \hline
$0.01$ & $0.26$ & $7$& $0.67$ & $0.684$ \\ \hline
$0.04$ & $0.06$ & $9$& $0.72$ & $0.737$ \\ \hline
\end{tabular}
\end{center}
\caption{Sets of values for which the assumptions on system parameters $(\alpha, \Delta, N_{min})$ and derived algorithm parameters $(\gamma, \beta)$ are satisfied.  }\label{table:assumptions}
\end{table}

\remove{
\begin{center}
\begin{tabular}{ |c|c|c|c|c| }
\hline
Churn Rate ($ \alpha$) & Failure Factor ($\Delta$) & Join Bound ($\gamma$) & Minimum System Size ($N_{min})$& Quorum Size ($\beta$) \\ \hline
$0.0001$ & $0.16$ & $0.679$ & $6$& $0.661$ \\ \hline
$0.001$ & $0.16$ & $0.674$ & $6$& $0.664$ \\ \hline
$0.01$ & $0.13$ & $0.68$ & $6$& $0.668$ \\ \hline
$0.02$ & $0.10$ & $0.69$ & $7$& $0.675$ \\ \hline
$0.03$ & $0.05$ & $0.7$ & $6$& $0.66$ \\ \hline
$0.04$ & $0.05$ & $0.68$ & $8$& $0.7$ \\ \hline

\end{tabular}
\end{center}
}

\remove{
\begin{align}
-1/\log_2(1- \alpha) &\geq 4 \label{parameters:G}
\end{align}}

\section{Correctness Proof} \label{section:proof}

We will show that {\sc \AlgName{}} satisfies the three properties listed
at the end of Section~\ref{section:model}.  
Lemmas~\ref{lem:size} through~\ref{lem:knowswhenjoined} 
are used to prove
Theorem~\ref{thm:joins}, 
which states that every node eventually joins, provided it does not
crash or leave.  Lemmas~\ref{lem:present-2D} through~\ref{lem:joined-for-D}
are used to prove
Theorem~\ref{thm:ops-live}, which states
that every operation invoked by a node that remains active eventually
completes.  Lemmas~\ref{lem:lin1} through~\ref{lem:lin4} are used to prove Theorem~\ref{thm:atomicity}, which
states that atomicity is satisfied.

Consider any execution.
We use $Changes_p^t$, $Present_p^t$, and $Members_p^t$ to denote the sets $Changes_p$, $Present_p$, and $Members_p$, respectively, at time $t$ of the execution. 
We begin by bounding the number of nodes that enter
during an interval of time
and the number of nodes that are present at the end of the
interval, as compared to the number present at the beginning.

\begin{lemma}\label{lem:size}
For all $i \in \nats$ and all $t \geq 0$, at most $((1 + \alpha)^i-1)N(t)$ nodes enter during $(t,t+Di]$ and $(1 - \alpha)^i N(t) \leq N(t+Di) \leq (1+\alpha)^iN(t)$.
\end{lemma}

\begin{proof}
The proof is by induction on $i$.
For $i = 0$ and all $t \geq 0$, $(t,t+Di]$ is empty,
and hence, $0 = ((1+\alpha)^i-1)N(t)$ nodes enter during this interval and
\[ N(t+iD) = N(t) =  (1+\alpha)^iN(t) = (1 - \alpha)^i N(t) . \]
Now let $i \geq 0$ and $t \geq 0$.
Suppose
at most $((1 + \alpha)^i-1)N(t)$ nodes enter during $(t,t+Di]$ and
$(1 - \alpha)^i N(t) \leq N(t+Di) \leq (1+\alpha)^iN(t)$.

Let $e\geq 0$ and $\ell\geq 0$ be the number of nodes that enter and leave,
respectively, during $(t+Di,t+D(i+1)]$.
By the churn assumption, $e+\ell\leq \alpha N(t+Di)$, so
$e,\ell \leq \alpha N(t+Di) \leq \alpha(1+\alpha)^iN(t)$.
The number of nodes that enter during $(t,t+D(i+1)]$ is at most
\[
((1 + \alpha)^i-1)N(t) + e
    \leq ((1 + \alpha)^i-1)N(t) + \alpha(1+\alpha)^iN(t)
    = ((1+\alpha)^{i+1}-1)N(t) . \]
Hence,
\[ N(t + D(i+1)) \leq N(t) +  ((1+\alpha)^{i+1}-1)N(t) = (1+\alpha)^{i+1}N(t) .\]
Furthermore,
\[ N(t + D(i+1)) \geq N(t+Di) -\ell
    \geq N(t+Di) - \alpha N(t+Di) = (1 - \alpha)N(t+Di)
    \geq (1-\alpha)^{i+1}N(t) . \]
By induction, the claim is true for all $i \in \nats$.
\end{proof}


We are also interested in the number of nodes that leave
during an interval of time.
In the proof of the next lemma, the calculation of the maximum number of nodes that leave during an interval
is complicated by the possibility of nodes entering during the interval,
allowing additional nodes to leave.

\begin{lemma}\label{lem:left2}
For $\alpha >0$, all nonegative integers $i \leq -1/\log_2(1-\alpha)$ and all $t \geq 0$, at most $(1-(1-\alpha)^i) N(t)$ nodes leave during $(t,t+Di]$.
\end{lemma}

\begin{proof}
The proof is by induction on $i$. When $i = 0$, the interval is empty,
so $0 = (1-(1-\alpha)^0) N(t)$ nodes leave during the interval.
Now let $i \geq 0$, let $t \geq 0$, and suppose
at most $(1-(1-\alpha)^i) N(t+D)$ nodes leave during $(t+D,t+D(i+1)]$.

Let $e\geq 0$ and $\ell\geq 0$ be the number of nodes that enter and leave, respectively, during $(t,t+D]$.
By the churn assumption, $e+\ell \leq \alpha N(t)$, so $\ell \leq  \alpha N(t)$ and
$N(t+D) = N(t)+ e - \ell  =  N(t) + (\ell + e) - 2\ell \leq (1 + \alpha) N(t) - 2\ell$.
The number of nodes that leave during $(t,t+D(i+1)]$
is the number that leave during $(t,t+D]$ plus the number that leave during $(t+D,t+D(i+1)]$, which
is at most
\begin{eqnarray*}
\ell + (1-(1-\alpha)^i) N(t+D)
&\leq & \ell + (1-(1-\alpha)^i) [(1 + \alpha) N(t) - 2\ell]\\
&= & (1-(1-\alpha)^i) (1 + \alpha) N(t) + ( 2(1-\alpha)^i-1) \ell\\
& \leq & (1-(1-\alpha)^i) (1 + \alpha) N(t) + ( 2(1-\alpha)^i-1)\alpha N(t)\\
& = & (1-(1-\alpha)^{i+1}) N(t).
\end{eqnarray*}
Note that $2(1-\alpha)^i -1 \geq 0$, since $i \leq -1/\log_2(1-\alpha)$.
By induction, the claim is true for all $i \in \nats$.
\end{proof}

Recall 
 that a node is {\em active} at time $t$ if it has entered by time $t$,
but has not left or crashed by time $t$.
The next lemma shows that some node remains active throughout any interval of length $3D$.

\begin{lemma}
\label{lem:gen0}
For every $t > 0$,
at least one node
is active throughout $[\max\{0,t -2D\},t +D]$.
\end{lemma}

\begin{proof}
Let $S$ be the set of nodes present at time
$t' = \max\{0,t -2D\}$,
so $|S| =  N(t') \geq N_{min}$.
By Lemma~\ref{lem:size}, at most $((1+\alpha)^3-1)|S|$ 
nodes enter during $(t',t+D]$,
so there are at most  $(1+\alpha)^3|S|$ nodes present at time $t+D$ and 
at most $ \Delta(1+\alpha)^3|S|$ nodes have crashed by time $t+D$.
Assumption~(\ref{parameters:G}) implies that
$-1/\log_2(1- \alpha) \geq 
4 \geq 
 3$. So,
by Lemma~\ref{lem:left2},
at most $(1-(1-\alpha)^3) |S|$ nodes leave during
$(t',t+D]$ 
and there are at least $(1-\alpha)^3 |S|$ nodes present at time $t+D$. 
Thus, at least
\begin{equation}
((1-\alpha)^3   - \Delta(1+\alpha)^3)|S|  \geq ((1-\alpha)^3   -\Delta(1+\alpha)^3 ) N_{min} 
\end{equation}
nodes in $S$ are active at time $t+D$.
By Assumption~(\ref{parameters:D}),
$ \left((1- \alpha)^3 -\Delta(1+\alpha)^3\right)N_{min}  > 1$,
so at least one node in $S$ is still active at time
$t+D$.
\end{proof}

\remove{\faith{enter, leave, and join events WERE NOT DEFINED, NOR WAS THE NOTATION $t_q^e$, $t_q^\ell$, and $t_q^j$.
I HAVE ADDED DEFINITIONS IN THE SECOND PARAGRAPH OF SECTION 3.}}

We define the set of all enter, join, and leave events that occur
during time interval $I$ to be
\[
\mbox{\it SysInfo}^I =
    \{ enter(q) ~|~t_q^e \in I\}
    \cup \{ join(q) ~|~t_q^j \in I\}
    \cup \{ leave(q) ~|~t_q^\ell \in I\} . \]
In particular,
$\mbox{\it SysInfo}^{[0,0]} = \{ enter(q) ~|~q \in S_0\}
                        \cup  \{ join(q) ~|~q \in S_0\}$.

Since a node $p$ that is active throughout $[t_p^e,t+D]$ directly
receives all enter, joined, and leave messages broadcast during $[t_p^e,t]$,
within $D$ time, we have:

\begin{observation}
\label{obs:S0}
For every node $p$ and all times $t \geq t_p^e$, if $p$ is active at time $t + D$, then\\
\knows{p}{t+D}{[t_p^e, t]}.
\end{observation}

By assumption, for every node $p \in S_0$,
$\mbox{\it SysInfo}^{[0,0]}  \subseteq \mbox{\it Changes}_p^0$,
and hence Observation~\ref{obs:S0} implies:

\begin{observation}
\label{obs:S1}
For every node $p \in S_0$, if $p$ is active at time $t \geq 0$, then\\
\knows{p}{t}{[0, \max\{0,t-D\}]}.
\end{observation}

The purpose of Lemmas~\ref{lem:gen1}, \ref{lem:gen2},
and~\ref{lem:knowswhenjoined}
is to show that information about
nodes entering, joining, and leaving is propagated properly, via the
{\em Changes} sets.

\begin{lemma}
\label{lem:gen1}
Suppose that, at time $T''$, a node $p \notin S_0$ receives an enter-echo message from a node $Q$ sent at time $T'$ in response to an enter message from $p$.
Let $T$ be any time such that $\max\{0,T''-2D\} \le T \le t_p^e$.
Suppose $p$ is active at time $T + 2D$ and
$Q$ is active throughout $[U,T+D]$, where $U \le \max\{0,T''-2D\}$.
Then {\em SysInfo}$^{(U,T]} \subseteq$ {\em Changes}$_p^{T+2D}$.

\remove{Suppose a node $p \notin S_0$ receives an enter-echo message at time
$t''$ from a node $Q$ that sends it at time $t'$ in response to an enter
message from $p$.
Let $t$ be any time such that $\max\{0,t''-2D\} \le T \le t_p^e$.
Suppose $p$ is active at time $T + 2D$ and
$q$ is active throughout $[u,t+D]$, where $u \le t $. 
Then {\em SysInfo}$^{(u,t]} \subseteq$ {\em Changes}$_p^{t+2D}$.}
\end{lemma}

\begin{proof}
Consider any node $r$ that enters, joins, or leaves at time $\hat{t} \in
(U,T]$.
Note that $Q$ 
directly receives the announcement of this event, 
since  $Q$ is active throughout $(U,T+D]$, which contains $[\hat{t},\hat{t}+D]$,
the maximum interval during which the announcement message is in transit.
We consider two cases, depending on
the time,  $v$, at which $Q$ receives this message. 
\begin{itemize}
\item[Case 1:] 
$v \leq T'$.
Since $Q$ receives the enter message from $p$ at $T'$,
information about this change to $r$ is in
{\em Changes}$_Q^{T'}$, which is part of the enter-echo message
that $Q$ sends to $p$ at time $T'$.
Thus, this information is in 
{\em Changes}$_p^{T''} \subseteq$ {\em Changes}$_p^{T+2D}$.
\item[Case 2:] 
$v > T'$. 
Messages are not received before they
are sent, so $T' \ge t_p^e$.
Since $v \le \hat{t} + D$, it follows that $v + D \le \hat{t} + 2D \leq T+2D$.
Thus $[v,v+D]$ is contained in $[t_p^e,T+2D]$. 
Immediately after receiving the announcement about $r$, node $Q$ broadcasts an echo message in response.
Since $p$ is active throughout  this interval,
it directly receives this echo message.
\end{itemize}
In both cases,
the information about $r$'s change reaches $p$ by time $T + 2D$.
it follows that {\em SysInfo}$^{(U,T]} \subseteq$ {\em Changes}$_p^{T+2D}$.
\remove{
Consider a node $r$ that enters, joins, or leaves at time $\hat{t} \in (u,t]$.
Note that $q$ receives the enter / join / leave
message from $r$ since $[\hat{t},\hat{t}+D]$,
the maximum interval during which the message is in transit, is contained
within $(u,t+D]$, during which $q$ is active.

If $q$ receives the message about this change from $r$ by time $t'$,
when $q$ receives the enter message from $p$,
then the change is in {\em Changes}$_q^{t'}$,
which is part of the enter-echo message that $q$ sends to $p$ at time $t'$.
Thus, the information is in
{\em Changes}$_p^{t''} \subseteq$ {\em Changes}$_p^{t+2D}$.

Otherwise, $q$ receives the message from $r$
after the enter message from $p$, say at time $v > t'$,
causing $q$ to send an echo.
We now show that $p$ is active throughout $[v,v+D]$,
ensuring that $p$ receives the echo message about $r$ from $q$.
Since $v > t'$ and messages are not received before they are sent,
$t' \ge t_p^e$.
Since $v \le \hat{t} + D$, it follows that $v + D \le \hat{t} + 2D$,
which is at most $t + 2D$ by choice of $\hat{t}$.
Thus $[v,v+D]$ is contained in $[t_p^e,t+2D]$ during which $p$ is active.

In both cases, $p$ receives the information about $r$'s change
by time $t + 2D$, and hence,
{\em SysInfo}$^{(u,t]} \subseteq$ {\em Changes}$_p^{t+2D}$.}
 \end{proof}

\begin{lemma}
\label{lem:gen2}
For every node $p$, if $p$ is active at time $t \geq t_p^e + 2D$,  then \knows{p}{t}{[0,t-D]}.

\end{lemma}

\begin{proof}
The proof is by induction on the order in which nodes enter the system.
If $p \in S_0$, then $t_p^e = 0$, so \knows{p}{t}{[0,t-D]} follows from Observation~\ref{obs:S1}.

Now consider any node $p \not\in S_0$ and suppose that the claim is true for all nodes that
enter earlier than $p$. Suppose $p$ is active at  time $t \geq t_p^e+ 2D$.
By Lemma~\ref{lem:gen0}, there is at least one node $q$ that is active throughout $[\max\{0,t_p^e -2D\},t_p^e +D]$.
Node $q$ receives an enter message from $p$ at some time $t' \in [t_p^e, t_p^e +D]$
and sends an enter-echo message back to $p$.
This message is received by $p$ at some time $t'' \in [t',t'+D]$.

If $q \in S_0$, then
\knows{q}{t'}{[0,\max\{0,t'-D\}]}, by Observation~\ref{obs:S1}.
If $q \not\in S_0$, then $0 < t_q^e \leq \max\{0,t_p^e -2D\}$, so
$t_q^e \leq  t_p^e - 2D$. Therefore
$t_q^e + 2D \leq t_p^e \leq t'$.
Since $q$ entered earlier than $p$, it follows from the induction hypothesis
that  \knows{q}{t'}{[0,t'-D]}.
Thus, in both cases, \knows{q}{t'}{[0,\max\{0,t'-D\}]}.
At time $t''\leq t$, $p$ receives the enter-echo message from $q$, so
\knows{p}{t''}{[0,\max\{0,t'-D\}]} $\subseteq \mbox{\it Changes}_p^{t}$.




Applying Lemma~\ref{lem:gen1} with $Q = q$, $U = \max\{0,t_p^e -D\}$, $T = t_p^e$, $T' =t'$ and $T'' = t''$ gives
{\em SysInfo}$^{(\max\{0,t'-D\},t_p^e]} \subseteq$ {\em Changes}$_p^{t_p^e + 2D}$. 
Since $t \geq t_p^e+2D$, $\mbox{\it Changes}_p^{t_p^e+2D}$ is a subset of $\mbox{\it Changes}_p^{t}$. Observation~\ref{obs:S0} implies \knows{p}{t}{[t_p^e,t-D]}. Hence,  \knows{p}{t}{[0,t-D]}.
\end{proof}

\begin{lemma}
\label{lem:knowswhenjoined}
 For every node $p \not\in S_0$, if $p$ joins at time $t_p^j$ and
is active at time $t \geq  t_p^j$, then \knows{p}{t}{[0,\max\{0,t-2D\}]}.
\end{lemma}

\begin{proof}
The proof is by induction on the order in which nodes join the system. Let $p \not\in S_0$ be a node that joins at time $t_p^j \leq t$ and suppose the claim holds for all nodes that join before $p$. If $t \geq t_p^e +2D$, then the claim follows by Lemma~\ref{lem:gen2}. So, assume that $t < t_p^e +2D$.

Before $p$ joins, it receives an enter-echo message from a joined node in response to its enter message. Suppose $p$ first receives such an enter-echo message at time $t''$ and this enter-echo was sent by $q$ at time $t'$. Then $t_p^e \leq t' \leq t'' \leq t_p^j$. 
If $q \in S_0$, then Observation \ref{obs:S1} implies that
\knows{q}{t'}{[0,\max\{0,t'-D\}]}. Otherwise, by the induction hypothesis,  \knows{q}{t'}{[0,\max\{0,t'-2D\}]}, 
since $q$ joined prior to $p$ and is active at time $t' \geq t_q^j$.
Note that  $\mbox{\it Changes}_q^{t'}$ 
 $\subseteq \mbox{\it Changes}_p^{t''} \subseteq \mbox{\it Changes}_p^{t} $.
If $t \leq 2D$, then $\max\{0,t-2D\} = 0$ and the claim is true. So, assume that $t > 2D$.

Let $S$ be the set of nodes present at time $\max\{0,t'-2D\}$, so  $|S|= N(\max\{0,t'-2D\})$. By Lemma~\ref{lem:left2} and Assumption~(\ref{parameters:G}), at most $(1-(1-\alpha)^3)|S|$ nodes leave during $(\max\{0,t'-2D\},t'+D]$. Since $t'' \leq t'+D$, it follows that $|\mbox{\it Present}_p^{t''}| \geq |S| - (1-(1-\alpha)^3)|S| = (1-\alpha)^3|S|$.
Hence, from lines~\ref{line:calculate join bound} 
and~\ref{line:check if enough enter echoes} 
of Algorithm~\ref{algo:Common}, $p$ waits until it has received at least $join\_bound = \gamma \cdot |\mbox{\it Present}_p^{t''}|  \geq \gamma \cdot (1-\alpha)^3|S| $ enter-echo messages before joining.

By Lemma~\ref{lem:size}, the number of nodes that enter during $(\max\{0,t'-2D\},t'+D]$ is at most $((1+\alpha)^3-1)|S|$.
Thus, at time $t'+D$, there are at most $ (1+\alpha)^3|S|$ nodes present and at most $ \Delta(1+\alpha)^3|S|$ nodes are crashed. 
\remove{\faith{WE KNOW THAT $t_p^e \leq t' \leq t_p^e +D$,  SO $t_p^e > t' -2D$ AND SINCE $p \not\in S_0$, $t_p^e > 0$.
THUS $p$ ENTERED DURING $[\max\{0,t'-2D\},t'+D]$, BUT
WAS NOT ACTIVE THROUGHOUT THIS INTERVAL. SO WHY WERE YOU SUBTRACTING AN EXTRA 1 FOR $p$?
PLEASE CHECK THIS.}}
Hence, the number of enter-echo messages $p$ receives before joining from nodes
that were active throughout $[\max\{0,t'-2D\},t'+D]$ is 
$join\_bound$ minus the total number of enters, leaves and crashes, 
which is at least 
\begin{multline}
\gamma \cdot (1-\alpha)^3 |S| - [((1+\alpha)^3 -1)|S| + (1-(1-\alpha)^3)|S| + \Delta(1+\alpha)^3 |S| ] \\
= [(1+\gamma) (1-\alpha)^3 - (1+\Delta)(1+\alpha)^3]|S| \\
\geq  [(1+\gamma) (1-\alpha)^3 - (1+\Delta)(1+\alpha)^3]N_{min} . \label{equation1}
\end{multline}
Rearranging Assumption~(\ref{parameters:H}), we get
$[(1+\gamma)(1- \alpha)^3 -(1+\Delta)(1+\alpha)^3N_{min}] \geq 1$, so
expression~(\ref{equation1}) is at least $1$. 
Hence $p$ receives an enter-echo message at some time $T'' \leq t_p^j$ from a node $q'$ that is active throughout  $[\max\{0,t'-2D\},t'+D] \supseteq [\max\{0,t'-2D\},t-D] $. Let $T'$ be the time that $q'$ sent its enter-echo message in response to the enter message from $p$. 
Applying Lemma~\ref{lem:gen1} with $Q = q'$, $U = \max\{0,t' -2D\}$, and $T = t-2D$
gives {\em SysInfo}$^{(\max\{0,t'-2D\},t-2D]} \subseteq$ {\em Changes}$_p^t$. 
Thus {\em SysInfo}$^{[0,t-2D]} =$ {\em SysInfo}$^{[0,\max\{0,t'-2D\}]} \cup$ {\em SysInfo}$^{(\max\{0,t'-2D\},t-2D]}
\subseteq$ {\em Changes}$_p^t$. 
\end{proof}

Next we prove that every node that remains active sufficiently long
after it enters succeeds in joining.

\begin{theorem}
\label{thm:joins}
Every node $p \not\in S_0$ that
is active at time $t_p^e + 2D$ joins by time $t_p^e + 2D$.
\end{theorem}

\begin{proof}
The proof is by induction on the order in which nodes enter the system. 
Let $p \not\in S_0$ be a node that enters at time $t_p^e$ and is active at time $t_p^e + 2D$.
Suppose the claim is true for all nodes that enter before $p$.

By Lemma~\ref{lem:gen0}, there is a node $q$ that is active throughout
$[\max\{t_p^e -2D,0\},t_p^e +D]$.
If $q \in S_0$, then $q$ joins at time 0.
If not, then 
$t_q^e 
\leq 
  t_p^e-2D$,
so, by the induction hypothesis,
$q$ joins by  time
$t_q^e + 2D 
\leq 
 t_p^e$.
Since $q$ is active at time  $t_p^e+D$,
it receives the enter message from $p$ during $[t_p^e, t_p^e+D]$
and sends an enter-echo message in response.
Since $p$ is active at time $t_p^e+2D$,
it receives the enter-echo message from $q$ by time $t_p^e+2D$.
Hence, by time $t_p^e+2D$, $p$ receives at least one enter-echo message
from a joined node in response to its enter message.

Suppose the first enter-echo message $p$ receives from a joined node in response to its enter message is sent by node $q'$ at time $t'$ and received by $p$ at time $t''$. By Lemma~\ref{lem:knowswhenjoined}, \knows{q'}{t'}{[0,\max\{0,t'-2D\}]} $\subseteq \mbox{\it Changes}_p^{t''}$.

Let $S$ be the set of nodes present at time $\max\{0,t'-2D\}$.
Since $t'' \leq t'+D$, it follows from Lemma~\ref{lem:size} that at most $( (1+\alpha)^3 -1)|S|$
nodes enter during $(\max\{0,t'-2D\},t'']$. Thus, $|\mbox{\it Present}_p^{t''}| \leq|S| +( (1+\alpha)^3 -1)|S|  = (1+\alpha)^3 |S|$. 
From line~~\ref{line:calculate join bound} in Algorithm~\ref{algo:Common}, it follows that
$join\_bound \leq \gamma \cdot (1 +\alpha)^3|S|$. 


By Lemma~\ref{lem:left2} and Assumption~(\ref{parameters:G}),
at most $(1-(1-\alpha)^3)|S|$ nodes leave during $(\max\{0,t'-2D\},t'+D]$. Also, by Lemma~\ref{lem:size}, at most $(1+\alpha)^3|S|$ nodes are present in the system at $t'+D$ and so at most $\Delta(1+\alpha)^3|S|$ nodes are crashed at $t'+D$. 
Since $t_p^e \leq t' \leq t_p^e +D$, 
the nodes in $S$ that do not leave during $(\max\{0,t'-2D\},t'+D]$ and are not crashed at $t'+D$ 
are active throughout $[t_p^e,t_p^e+D]$ and send enter-echo messages in response to $p$'s enter message.
By time $t_p^e+2D$,  $p$ receives all these enter-echo messages.
There are at least $|S| - (1- (1-\alpha)^3) |S|-\Delta(1+\alpha)^3|S|= (1-\alpha)^3 |S|-\Delta(1+\alpha)^3|S|$
such enter-echo messages. 
By Assumption~(\ref{parameters:B}),
\begin{align*}
\frac{(1-\alpha)^3}{(1+\alpha)^3}-\Delta &\geq \gamma,
\end{align*}
so the value of  $join\_bound$ is 
at most 
$$\gamma \cdot (1 +\alpha)^3|S| 
\leq ( \frac{(1-\alpha)^3}{(1+\alpha)^3}-\Delta)  \cdot (1 +\alpha)^3|S|
= (1-\alpha)^3 |S|-\Delta(1+\alpha)^3|S|.$$
%
Thus, 
by time $t_p^e+2D$, 
 the condition in line~\ref{line:check if enough enter echoes} of Algorithm~\ref{algo:Common}
holds and node $p$ joins.
\end{proof}

Next, we show that all read and write operations terminate.
Specifically, we 
show that the number of responses for which an operation
waits is 
at most the number 
that it is guaranteed to receive.

Since $enter(q)$ is added to $\mbox{\it Changes}_p$ whenever $join(q)$
is, we get the following observation.

\begin{observation}
\label{obs:S2}
For every time $t \geq 0$ and  every node $p$ that is active at time $t$,
$\mbox{\it Members}_p^t \subseteq \mbox{\it Present}_p^t$.
\end{observation}


Lemma \ref{lem:present-2D} relates
a node's current estimate of the number of nodes present to
the number of nodes that were present in the
system $2D$ time units earlier. 
Lemma \ref{lem:members-2D} relates
a node's current estimate of the number of nodes that are members to
the number of nodes that were present in the
system time $4D$ time units earlier. 
Lemma \ref{lem:present-2D} is used in the proof of Lemma~\ref{lem:joined-for-D}  and Lemma~\ref{lem:members-2D} is used in the proof of Theorem~\ref{thm:atomicity}. The proofs of Lemmas~\ref{lem:members-2D} and  \ref{lem:present-2D} are very similar to each other and are thus presented together. 


\begin{lemma}
\label{lem:present-2D}
For every node $p$ and every time $t\geq t_p^j$ at which $p$ is active,
\begin{align*}
(1-\alpha)^2 \cdot N(\max\{0,t-2D\}) \le
|\mbox{\it Present}_p^t| \le (1+\alpha)^2 \cdot N(\max\{0,t-2D\}).
\end{align*}
\end{lemma}

\begin{proof}
By Lemma~\ref{lem:knowswhenjoined}, \knows{p}{t}{[0,\max\{0,t-2D\}]}.
Thus {\em Present}$_p^t$ contains all nodes that are present at
time $\max\{0,t-2D\}$,
plus any nodes that  enter in $(\max\{0,t-2D\},t]$ which $p$ has learned
about,
minus any nodes that  leave in $(\max\{0,t-2D\},t]$ which $p$ has learned about.
Then, by Lemma~\ref{lem:size},
$|Present_p^t| \leq N(\max\{0,t-2D\}) +( (1+\alpha)^2-1) \cdot N(\max\{0,t-2D\})
= (1+\alpha)^2 \cdot N(\max\{0,t-2D\})$.
Similarly, by Lemma~\ref{lem:left2} and Assumption~(\ref{parameters:G}),
$|Present_p^t| \geq N(\max\{0,t-2D\}) -(1- (1-\alpha)^2) \cdot N(\max\{0,t-2D\})
= (1-\alpha)^2 \cdot N(\max\{0,t-2D\})$.
\end{proof}

\remove{  
THE FOLLOWING IS A STRONGER VARIANT OF THE PREVIOUS LEMMA,
WHEN THE LOWER BOUND ON $t$ IS LARGER.

\begin{lemma}
\label{lem:present-D}
For every node $p$ and every time $t\geq0$ at which $p$ is active,
if $p \in S_0$ or $t\geq t_p^e+2D$, then
\begin{align*}
(1-\alpha) \cdot N(\max\{0,t-D\}) \le
|\mbox{\it Present}_p^t| \le (1+\alpha) \cdot N(\max\{0,t-D\}).
\end{align*}
\end{lemma}

\begin{proof}
Suppose $p \in S_0$ or $t\geq t_p^e+2D$.
By Observation~\ref{obs:S1} or Lemma~\ref{lem:gen2}, \knows{p}{t}{[0,\max\{0,t-D\}]}.
Thus {\em Present}$_p^t$ contains all nodes that are present at
time $t-D$,
plus any nodes that entered in $(t-D,t]$ which $p$ has learned
about,
minus any nodes that left in $(t-D,t]$ which $p$ has learned about.
Then, by Lemma~\ref{lem:size},
$|Present_p^t| \leq N(t-D) +\alpha\cdot N(t-D)
= (1+\alpha) \cdot N(t-D)$.
Similarly, by Lemma~\ref{lem:left2} and Assumption~(\ref{parameters:G}),
$|Present_p^t| \geq N(t-D) -\alpha \cdot N(t-D)
= (1-\alpha) \cdot N(t-D)$.
\end{proof}
} 

\begin{lemma}
\label{lem:members-2D}
For every node $p$ and every time $t\geq t_p^j$ at which $p$ is active,
\begin{align*}
(1-\alpha)^4 \cdot N(\max\{0,t-4D\}) \le
|\mbox{\it Members}_p^t| \le (1+\alpha)^4 \cdot N(\max\{0,t-4D\}).
\end{align*}
\end{lemma}

\begin{proof}
By Lemma~\ref{lem:knowswhenjoined}, \knows{p}{t}{[0,\max\{0,t-2D\}]} and, by
Theorem~\ref{thm:joins}, every node that enters by time $\max\{0,t-4D\}$
joins by time $\max\{0,t-2D\}$ if it is 
still 
active.
Thus {\em Members}$_p^t$ contains all nodes that are present at time $\max\{0,t-4D\}$
plus any nodes that enter in $(\max\{0,t-4D\},t]$ which $p$ learns have joined,
minus any nodes that leave in $(\max\{0,t-4D\},t]$ which $p$ 
learns have left. 
Then, by Lemma~\ref{lem:size},
$|Members_p^t| \leq N(\max\{0,t-4D\}) +( (1+\alpha)^4-1) \cdot N(\max\{0,t-4D\})
= (1+\alpha)^4 \cdot N(\max\{0,t-4D\})$.
Similarly, by Lemma~\ref{lem:left2} and Assumption~(\ref{parameters:G}),
$|Members_p^t| \geq N(\max\{0,t-2D\}) -(1- (1-\alpha)^4) \cdot N(\max\{0,t-4D\})
= (1-\alpha)^4 \cdot N(\max\{0,t-4D\})$.
\end{proof}

The next lemma
provides 
a lower bound on the number of nodes that will
reply to an operation's query or update message.

\begin{lemma}
\label{lem:joined-for-D}
If node $p$ is active at time $t \geq t_p^j$, then the number of nodes that
join by time $t$ and are still active at time $t+D$ is at least
$\left[\frac{(1-\alpha)^3}{(1+\alpha)^2} - \Delta(1+\alpha)\right]\cdot|\mbox{\it Present}_p^t| $.
\end{lemma}

\remove{\faith{I CHANGED $t'$ BACK TO $\max\{0,t-2D\}$ THROUGH OUT THIS PROOF.
WITH $t'$ THE RELATIONSHIP TO $t$ IS LOST.}}

\begin{proof}
By Lemma~\ref{lem:left2} and Assumption~(\ref{parameters:G}),
the maximum number of nodes that 
leave during
$(\max\{0,t-2D\},t+D]$
is at most $(1-(1-\alpha)^3) \cdot N(\max\{0,t-2D\})$.
By Lemma~\ref{lem:size}, at most
$ ((1+\alpha)^3-1)\cdot N(\max\{0,t-2D\})$
nodes 
enter during $(\max\{0,t-2D\},t+D]$. So, at most $ \Delta(1+\alpha)^3\cdot N(\max\{0,t-2D\})$ nodes are crashed by $t+D$.
Thus, there are at least
$$N(\max\{0,t-2D\}) - (1-(1-\alpha)^3) \cdot N(\max\{0,t-2D\})-  \Delta(1+\alpha)^3\cdot N(\max\{0,t-2D\})$$  $$ = [(1-\alpha)^3 - \Delta(1+\alpha)^3] \cdot N(\max\{0,t-2D\})$$
nodes that were present at time
$\max\{0,t-2D\}$
and are still active at time $t+D$.
This number is bounded below by
$\left[\frac{(1-\alpha)^3}{(1+\alpha)^2} - \Delta(1+\alpha)\right]\cdot |\mbox{\it Present}_p^t| $
since, by Lemma~\ref{lem:present-2D},
$N(\max\{0,t-2D\}) \ge |\mbox{\it Present}_p^t|/(1+\alpha)^2$.
By Theorem~\ref{thm:joins}, all of these nodes are joined by time $t$.
\end{proof}

\begin{theorem}
\label{thm:ops-live}
Every read or write operation invoked by a node that remains active completes.
\end{theorem}

\begin{proof}
Each operation consists of a read phase and a write phase.
We show that each phase terminates
within $2D$ time, provided the client does not crash or leave.

Consider a phase of an operation by client $p$ that starts at time $t$.
Every node that joins by time $t$ and is still active at time $t+D$
receives $p$'s query or update message
and replies with a response or ack message
by time $t+D$.
By Lemma~\ref{lem:joined-for-D}, there are at least
$\left[\frac{(1-\alpha)^3}{(1+\alpha)^2} - \Delta(1+\alpha)\right]\cdot |\mbox{\it Present}_p^t| $
such nodes.

From 
Assumption~(\ref{parameters:C}) and Observation~\ref{obs:S2},
\begin{align*}
\left[\frac{(1-\alpha)^3}{(1+\alpha)^2} - \Delta(1+\alpha)\right]\cdot |\mbox{\it Present}_p^t| \geq
           \beta |\mbox{\it Present}_p^t|
  \geq \beta |\mbox{\it Members}_p^t|  = quorum\_size_p^t.
\end{align*}
Thus, by time $t+2D$, $p$ receives sufficiently many response or ack messages
to complete the phase.
\end{proof}


\remove{
A write operation $W$ 
by node $p$ consists of a read phase $r$ followed by a write phase $w$. Let $t_w$ denote the time at the beginning of its write phase. At time $t_w$, node $p$ broadcasts an update message (on line~\ref{line:bcast update1} or line~\ref{line:bcast update2} of Algorithm~\ref{algo:rw_client}) containing a triple $(v,s,i)$, where $value(W) = v$ is the {\em value written by} $W$
and $ts(w) = (s,i) = (seq_p^{t_w}, id_p^{t_w})$
is the {\em timestamp} of $w$.
}

Now we prove atomicity of the \AlgName{} algorithm. 
Let $\cal T$ be the set of read operations that complete and write operations that execute line~\ref{line:bcast update1} of Algorithm~\ref{algo:rw_client}.
For any node $p$, let $ts_p^t =(seq_p^t,id_p^t)$ denote the {\em timestamp} of 
the latest register value known to 
 node $p$ at time $t$.
Note that 
new 
 timestamps are created by write operations
(on lines~\ref{line:new timestamp}-\ref{line:new timestamp2}  of Algorithm~\ref{algo:rw_client})
and are sent via enter-echo, update, and update-echo messages.
Initially, $ts_p^0 = (0,\bot)$ for all nodes $p$.

For any read or write operation $o$ in $\cal T$ by $p$, 
 the {\em timestamp of its read phase},
$ts^{rp}(o)$, is
$ts_p^t$, where $t$ is 
the end of its read phase
(i.e., when the condition on line~\ref{line:quorum reached} of
Algorithm~\ref{algo:rw_client} evaluates to true).
The {\em timestamp of its write phase},
$ts^{wp}(o) $, is $ts_p^t$, where $t$ is  the beginning of its write phase
(i.e., when it broadcasts on line~\ref{line:bcast update1} 
 of Algorithm~\ref{algo:rw_client}). 
The {\em timestamp of a read operation}  in $\cal T$ is the timestamp of its read phase.
The {\em timestamp of a write operation}  in $\cal T$ is the timestamp of its write phase. 

Lemmas \ref{lem:lin1}--\ref{lem:lin4}
show that write phase information propagates properly through the system.
They are analogous to Observation \ref{obs:S1} and Lemmas \ref{lem:gen1}--\ref{lem:knowswhenjoined},
which concern the propagation of information about {\em enter}, {\em join}, and {\em leave} events. 

\begin{lemma}
\label{lem:lin1}
If $o$ is an operation in $\cal T$ whose write phase w starts at $t_w$,
node $p$ is active at time $t \geq t_w+D$, and
$t_p^e \leq t_w$, then $ts_p^t \geq ts^{wp}(o)$.
\end{lemma}

\begin{proof}
Since $p$ is active throughout $[t_w,t_w+D]$, it directly receives the update message
broadcast by $w$ at time $t_w$.
Hence, from lines~\ref{line:new value?}--\ref{line:update value}
of Algorithm~\ref{algo:Server},
$ts_p^t \geq ts^{wp}(o)$.
\end{proof}

\begin{lemma}
\label{lem:lin2}
Suppose a node $p \not\in S_0$ receives an enter-echo message at time $t''$ from
a node $q$ that sends it at time $t'$ in response to an enter message from $p$.
If $o$ is an operation 
 whose write phase w starts at $t_w$,
$p$ is active at time $t \geq \max\{t'',t_w+2D\}$,
and $q$ is active throughout $[t_w,t_w + D]$,
then $ts_p^{t} \geq ts^{wp}(o)$.
\end{lemma}

\begin{proof}
Since $q$ is  active throughout
$[t_w,t_w + D]$, it receives the update message from $w$ at some time $\hat{t} \in [t_w,t_w + D]$,
so $ts_q^{\hat{t}} \geq ts^{wp}(o)$.
At time $t'' \leq t$, node $p$ receives the enter-echo sent by node $q$ at time $t'$,
so $ts_p^t \geq ts_p^{t''} \geq ts_q^{t'}$.
If $t' \geq \hat{t}$, then $ts_q^{t'} \geq ts_q^{\hat{t}}$, so
$ts_p^{t} \geq ts^{wp}(o)$.
If $\hat{t} > t'$, then $q$ sends an update-echo at time $\hat{t} \leq t_w+D$,
$p$ receives it by time $\hat{t}+D  \leq t_w+2D \leq t$, and, thus,
$ts_p^t \geq ts_q^{\hat{t}} \geq ts^{wp}(o)$.
\end{proof}

\remove{
If $q$ receives the update message from $w$
before the enter message from $p$, then
$ts_q^{t'} \geq ts(w)$ from Lines~\ref{line:new value?}-\ref{line:update value}
in Algorithm~\ref{algo:Server}.
By Lines 5-6 in Algorithm 1, $ts_p^{t''} \geq ts_q^{t'}$.
Since $t \geq t''$, if follows that $ts_p^{t} \geq ts(w)$.

Otherwise,  $q$ receives the update message from $w$
after the enter message from $p$ and sends an update-echo message
in response by time $t_w+D$.
Since $p$ receives this message from $q$ by time $t_w+2D \leq t$, it follows that
$ts_p^{t} \geq ts_p^{t_w+2D} \geq  ts(w)$.
} 

\begin{lemma}
\label{lem:lin3}
If $o$ is an operation in $\cal T$ whose write phase w starts at $t_w$ and
node $p$ is active at time $t \geq \max\{t_p^e+2D,t_w+D\}$,
then $ts_p^t \geq ts^{wp}(o)$.
\end{lemma}

\begin{proof}
The proof is by induction on the order in which nodes enter the system.
Suppose the claim is true for all nodes that enter earlier than $p$.
If $t_p^e \leq t_w$, which is the case for all $p \in S_0$,
then the claim follows from Lemma~\ref{lem:lin1}.
So, suppose that $t_w < t_p^e$.

By Lemma~\ref{lem:gen0}, there is at least one node $q$ that is active throughout
$[\max\{0,t_p^e-2D\},t_p^e+D]$.
It receives an enter message from $p$ at some time $t' \in [t_p^e, t_p^e +D]$
and sends an enter-echo message containing $ts_q^{t'}$ back to $p$.
This message is received by $p$ at some time $t'' \leq t'+D \leq t_p^e +2D \leq t$, 
so 
$ts_q^{t'} \leq ts_p^{t''} \leq ts_p^t$.

The first case is when 
 $t_w \geq \max\{0,t_p^e-2D\}$.
Since $t_w + D < t_p^e +D$, it follows that $q$ is active throughout $[t_w,t_w+D]$. Furthermore,
$t \geq t_p^e+2D \geq \max\{t'', t_w+2D\}$.
Hence, Lemma~\ref{lem:lin2} implies that
$ts_p^{t} \geq ts^{wp}(o)$.

The second case is when 
 $t_w < \max\{0,t_p^e-2D\}$. 
 Since $t_w \geq 0$, it follows that
$t_p^e-2D > 0$, $t_q^e \leq \max\{0,t_p^e -2D\} = t_p^e -2D$,  and $t_w < t_p^e-2D \leq t'-2D$,
so $t' \geq \max\{t_q^e + 2D, t_w + D\}$.
Note that $q$ is active at time $t'$ and $q$ enters before node $p$, so, by the induction hypothesis,
$ts_q^{t'} \geq ts^{wp}(o)$. 
Hence, 
$ts_p^t \geq ts^{wp}(o)$.
\end{proof}

\begin{lemma}
\label{lem:lin4}
If $o$ is an operation in $\cal T$ whose write phase starts at $t_w$,
node $p \not\in S_0$ joins at time $t_p^j$, and
$p$ is active at time $t \geq \max\{t_p^j,t_w+2D\}$,
then $ts_p^t \geq ts^{wp}(o)$.
\end{lemma}

\begin{proof}
The proof is by induction on the order in which nodes enter the system.
Suppose the claim is true
for all nodes that
join before $p$. If $t \geq t_p^e +2D$, then the claim follows by Lemma~\ref{lem:lin3}. So, assume
that $t < t_p^e +2D$.
If $t_p^e \leq t_w$, then the claim follows by Lemma~\ref{lem:lin1}. So, assume
that $t_w < t_p^e$.

Before $p$ joins, it receives an enter-echo message from a joined node in response to its enter message.
Suppose $p$ first receives such an enter-echo message
at time $t''$ and this enter-echo was sent by $q$ at time $t'$. Then $t'' \leq t_p^j \leq t$
and $ts_q^{t'} \leq ts_p^{t''} \leq ts_p^t$.

\remove{\faith{YOU CAN'T SAY ``FROM THE PROOF OF LEMMA''. YOU HAVE TO RESTATE THE LEMMA
SO WHAT YOU WANT IS IN THE STATEMENT, OR REPROVE IT HERE.}}

Now we prove that $p$ receives an enter echo message from a node $q'$ that is active throughout $[\max\{0,t'-2D\}, t' +D]$.
Let $S$ be the set of nodes present at time $\max\{0,t'-2D\}$, so  $|S|= N(\max\{0,t'-2D\})$. By Lemma~\ref{lem:left2} and Assumption~(\ref{parameters:G}), at most $(1-(1-\alpha)^3)|S|$ nodes leave during $(\max\{0,t'-2D\},t'+D]$. Since $t'' \leq t'+D$, it follows that $|\mbox{\it Present}_p^{t''}| \geq |S| - (1-(1-\alpha)^3)|S| = (1-\alpha)^3|S|$.
Hence, from lines~\ref{line:calculate join bound} 
and~\ref{line:check if enough enter echoes} 
of Algorithm~\ref{algo:Common}, $p$ waits until it has received at least $join\_bound = \gamma \cdot |\mbox{\it Present}_p^{t''}|  \geq \gamma \cdot (1-\alpha)^3|S| $ enter-echo messages before joining.

By Lemma~\ref{lem:size}, the number of nodes that enter during $(\max\{0,t'-2D\},t'+D]$ is at most $((1+\alpha)^3-1)|S|$.
Thus, at time $t'+D$, there are at most $ (1+\alpha)^3|S|$ nodes present and at most $ \Delta(1+\alpha)^3|S|$ nodes are crashed. Hence, the number of enter-echo messages $p$ receives before joining from nodes
that were active throughout $[\max\{0,t'-2D\},t'+D]$ is 
$join\_bound$ minus the total number of enters, leaves and crashes, 
which is at least 
\begin{multline}
\gamma \cdot (1-\alpha)^3 |S| - [((1+\alpha)^3 -1)|S| + (1-(1-\alpha)^3)|S| + \Delta(1+\alpha)^3 |S| ] \\
\geq  [(1+\gamma) (1-\alpha)^3 - (1+\Delta)(1+\alpha)^3]N_{min} . \label{equation1}
\end{multline}
Rearranging Assumption~(\ref{parameters:H}), we get
$[(1+\gamma)(1- \alpha)^3 -(1+\Delta)(1+\alpha)^3N_{min}] \geq 1$, so
expression~(\ref{equation1}) is at least $1$.  Hence $p$ receives an enter-echo message at some time $T'' \leq t_p^j$ from a node $q'$ that is active throughout  $[\max\{0,t'-2D\},t'+D] \supseteq [\max\{0,t'-2D\},t-D] $. Let $T'$ be the time that $q'$ sent its enter-echo message in response to the enter message from $p$.
Then $ts_{q'}^{T'} \leq ts_p^{T''} \leq ts_p^t$.

Note that $t_w < t_p^e \leq t'$ so $t_w + D \leq t'+D$.
If $t_w \geq \max\{0,t'-2D\}$, then $q'$ is active throughout $[t_w,t_w+D]$.
Since $t \geq \max\{T'',t_w+2D\}$, it follows by Lemma~\ref{lem:lin2} that $ts_p^t \geq ts^{wp}(o)$.
So, assume that $t_w < \max\{0,t'-2D\}$.

Since
$t_w 
\geq 
 0$,
it follows that $t' > t_w + 2D$.
If $q \in S_0$, then $t_q^e = 0
\leq 
 t_w$, 
 so, by Lemma~\ref{lem:lin1},
$ts_q^{t'} \geq ts^{wp}(o)$.
If $q \not\in S_0$, then, by the induction hypothesis,
$ts_q^{t'} \geq ts^{wp}(o)$, since $q$ joins at time $t_q^j < t_p^j \leq t'$.
Thus, in both cases, $ts_p^{t} \geq ts^{wp}(o)$.
\end{proof}

\remove{Given an execution, we order all the read operations that complete and all write operations that execute Line~\ref{line:bcast update1} of Algorithm~\ref{algo:rw_client} as follows.
First, order the write operations in order of their timestamps.
Note that all write operations have different timestamps, since each write operation by node $p$
has a timestamp with second component $p$ and first component larger than any timestamp
$p$ has previously seen.
Then insert each read operation immediately following the write operation
with the same timestamp.
Break ties among read operations by the order in which they start.}

Lemma~\ref{lem:lin5} is the key lemma for proving atomicity of \AlgName{}. It shows that for two non-overlapping operations in $\cal T$, the timestamp of the read phase of the latter operation is at least as big as the timestamp of the write phase of the former. Lemma~\ref{lem:lin6} uses Lemma~\ref{lem:lin5} to show that the timestamps of two non-overlapping operations respect real time ordering. Theorem~\ref{thm:atomicity} uses Lemmas~\ref{lem:lin5} and \ref{lem:lin6} to complete the proof of atomicity.

\begin{lemma}\label{lem:lin5}
For any two operations $op_1$ and $op_2$ in $\cal T$, 
 if $op_1$ finishes before $op_2$ starts, then $ts^{wp}(op_1) \leq ts^{rp}(op_2)$.
\end{lemma}

\begin{proof}

Let $p_1$ be the node that invokes $op_1$, let $w$ denote the write phase of $op_1$,
let $t_w$ be the start time of $w$, and let $\tau_w = ts^{wp}(op_1) = ts^{t_w}_{p_1}$.
Similarly, let $p_2$ be the node that invokes $op_2$, let $r$ denote the read phase of $op_2$,
let $t_r$ be the start time of $r$, and let $\tau_r = ts^{rp}(op_2) = ts^{t_r}_{p_2}$.

Let $Q_w$ be the set of nodes that $p_1$ hears from during $w$ (i.e., that sent messages causing $p_1$ to increment $heard\_from$ on line~\ref{line:inc heard from ack} of Algorithm~\ref{algo:rw_client}) and $Q_r$ be the set of nodes that $p_2$ hears from during $r$ (i.e., that sent messages causing $p_2$ to increment $heard\_from$ on line~\ref{line:inc heard from resp} of Algorithm~\ref{algo:rw_client}). Let $P_w = |Present_{p_1}^{t_w}| $ and $M_w = |Members_{p_1}^{t_w}| $ be the sizes of the {\em Present} and {\em Members} sets belonging to $p_1$ at time $t_w$, and $P_r = |Present_{p_2}^{t_r}|  $ and $M_r  = |Members_{p_2}^{t_r}| $ be the sizes of the {\em Present} and {\em Members} sets belonging to $p_2$ at time $t_r$. 
\medskip

\noindent
{\bf Case I:} $t_r > t_w + 2D$.

We start by showing there exists a node $q$ in $Q_r$ such that $t_q^j \leq t_r-2D$. Each node $q \in Q_r$ receives and responds to $r$'s query, so it 
joined by time $t_r+D$. By Theorem~\ref{thm:joins}, the number of nodes that can join during $(t_r - 2D, t_r + D]$ 
is at most the number of nodes that can enter in $(\max\{0,t_r - 4D\}, t_r + D]$. By Lemma~\ref{lem:size}, the number of nodes that can enter during $(\max\{0,t_r - 4D\}, t_r + D]$ is at most $((1 + \alpha)^5 - 1) \cdot N(\max\{0,t_r - 4D\})$. By Lemma~\ref{lem:members-2D}, $N(\max\{0,t_r - 4D\}) \le M_r/(1-\alpha)^4$. From the code and Assumption~(\ref{parameters:E}), it follows that $|Q_r| \ge \beta M_r  > M_r (1 + \alpha)^5 - 1)/(1 - \alpha)^4 \geq (1 + \alpha)^5 - 1) \cdot N(\max\{0,t_r - 4D\})$, which is at most the number of nodes that can enter in $(\max\{0,t_r - 4D\}, t_r + D]$.
Thus, there is a node $q \in Q_r$ that joins by time $t_r-2D$.

Suppose $q$ receives $r$'s query message at time $t' \geq t_r \geq t_w + 2D $ If $q \in S_0$, then $t_q^j = 0 \leq t_w$, so, by Lemma~\ref{lem:lin1}, $ts_q^{t'} \ge ts^{wp}(op_1) =  \tau_w$. Otherwise,  $q \not\in S_0$, so $0 < t_q^j \leq t_r - 2D < t'$. Since $t_w+2D < t_r \leq t'$, Lemma~\ref{lem:lin4} implies that $ts_q^{t'} \ge  ts^{wp}(op_1) =  \tau_w$. In either case, 
$q$ responds to $r$'s query message with a timestamp at least as large as $\tau_w$ and, hence, $\tau_r \geq \tau_w$.

\medskip
\noindent
{\bf Case II:} $t_r \le t_w + 2D$.

Let 
$J = \{ p~|~t_p^j < t_r$ and $p$ is active at time $t_r\} \cup \{p~|~ t_r \leq t_p^j \leq t_r+D\}$, which contains the set of all nodes that reply to $r$'s query. By Theorem~\ref{thm:joins}, all nodes that are present at time
$\max\{0,t_r-2D\}$ join by time $t_r$ if they remain active.
Therefore all nodes in $J$ are either active at time $\max\{0,t_r-2D\}$
or enter during $(\max\{0,t_r-2D\},t_r+D]$.
By Lemma~\ref{lem:size}, $|J| \leq (1+\alpha)^3N(\max\{0,t_r-2D\})$.

Let $K$ be the set of all nodes that are present at time $\max\{0,t_r-2D\}$ and do not leave or crash during $(\max\{0,t_r-2D\},t_r+D]$.
Note that $K$ contains all the nodes in $Q_w$ that do not leave or crash during $[t_w, t_r + D] \subseteq [\max\{0,t_r -2D\},t_r+D]$.
By Lemma~\ref{lem:left2} and Assumption~(\ref{parameters:G}), at most $(1-(1-\alpha)^3)N(\max\{0,t_r-2D\})$ nodes leave during $[\max\{0,t_r -2D\},t_r+D]$. By Lemma~\ref{lem:size}, at most $ ((1+\alpha)^3-1) N(\max\{0,t_r-2D\})$ nodes enter during $[\max\{0,t_r -2D\},t_r+D]$. So, at most $ \Delta(1+\alpha)^3  N(\max\{0,t_r-2D\})$ nodes are crashed at $t_r+D$. 

From the code, $|Q_r| \geq \beta M_r $ and, by Lemma~\ref{lem:members-2D}, $M_r \geq (1-\alpha)^4  N(\max\{0,t_r-4D\})$. So,
\begin{align*}
|Q_r| \geq \beta (1-\alpha)^4  N(\max\{0,t_r-4D\}).
\end{align*}
Similarly,
\begin{align*}
|Q_w| \geq \beta M_w  \geq \beta (1-\alpha)^4  N(\max\{0,t_w-4D\}).
\end{align*}
Therefore, by the definition of $K$:
\remove{\faith{IT IS NOT THE DEFINITION OF $K$, BUT THE SENTENCE FOLLOWING THE DEFINITION OF $K$.
THIS COMMENT WOULD BE FINE IF $K$ WAS DEFINED TO BE THE SET OF ALL NODES THAT ARE PRESENT AT TIME
$\max\{0,t_r-2D\}$ AND DO NOT LEAVE OR CRASH DURING $(\max\{0,t_r-2D\},t_r+D]$.}}
\begin{align*}
|K| &\ge |Q_w| -(1-(1-\alpha)^3 + \Delta  (1+\alpha)^3)N ( \max\{0,t_r-2D\})
\end{align*}
\begin{align}
    &\ge (\beta (1-\alpha)^4 N(\max\{0,t_w-4D\})) - (1-(1-\alpha)^3 + \Delta (1+\alpha)^3) N ( \max\{0,t_r-2D\}). \label{inequality:k}
\end{align}
Since $t_r - t_w < 2D$, it follows that
$\max\{0,t_r - 4D\} - \max\{0,t_w - 4D\} < 2D$.
By Lemma~\ref{lem:size}, $N(\max\{0,t_r - 4D\}) \le (1+\alpha)^2 \cdot N(\max\}0,t_w - 2D\})$.
Thus we can replace $N(\max\{0,t_w - 4D\})$ in Formula (\ref{inequality:k}) with $(1+\alpha)^2 \cdot N(\max\{0, t_r - 4D\})$ and get:

\remove{Since $t_w < t_r \leq t_w + 2D$, it implies that 
$t_w-4D < t_r -4D\leq t_w + 2D-4D$. 
Since $t_r-4D$ is in the interval $(t_w-4D, t_w-2D]$,  the system size at $t_r-4D$, $N(\max\{0,t_r-4D\}) $ is upper bounded by the maximum system size at $t_w-2D$. By Lemma~\ref{lem:size}, we know that the system size at $t_w-2D$ is upper bounded by the inequality $N(\max\{0,t_w-2D\}) \leq (1+\alpha)^2 N(\max\{0,t_w-4D\}) $. Thus, using the above two relations we get $N(\max\{0,t_w-4D\}) \geq (1+\alpha)^{-2}N(\max\{0,t_r-4D\})$.  Thus,}

\begin{align*}
|Q_r|+ |K|
&\geq
   \beta (1-\alpha)^4 N(\max\{0,t_r-4D\}) \\
&\hspace*{.3in} +\beta (1-\alpha)^4 (1+\alpha)^{-2}
     N(\max\{0,t_r-4D\}) \\
&\hspace*{.3in}- (1-(1-\alpha)^3+ \Delta(1+\alpha)^3)N(\max\{0,t_r-2D\})\\
&= \beta (1-\alpha)^4 (1+\alpha)^{-2} (2+2\alpha+\alpha^2) N(\max\{0,t_r-4D\}) \\
&\hspace*{.3in}-(\Delta(1+\alpha)^3 - (1-\alpha)^3 + 1 )N(\max\{0,t_r-2D\}).
\end{align*}
 By Lemma~\ref{lem:size},
$N(\max\{0,t_r-4D\}) \geq 
(1-\alpha)^{-2}N(\max\{0,t_r-2D\})$. Thus,
\begin{align*}
|Q_r|+ |K|
&\geq \beta (1-\alpha)^2 (1+\alpha)^{-2} (2+2\alpha+\alpha^2) N(\max\{0,t_r-2D\})\\
&\hspace*{.3in}-(\Delta(1+\alpha)^3 - (1-\alpha)^3 + 1)N(\max\{0,t_r-2D\})\\
&= (\beta  (1-\alpha)^2 (1+\alpha)^{-2} (2+2\alpha+\alpha^2) -(\Delta(1+\alpha)^3 - (1-\alpha)^3 + 1))N(\max\{0,t_r-2D\}) \\
\end{align*}
By Assumption~(\ref{parameters:F}),
$\beta(1-\alpha)^2 (1+\alpha)^{-2} (2+2\alpha+\alpha^2) > (1+\Delta)(1+\alpha)^3 - (1-\alpha)^3 + 1$, so
\begin{align*}
{\rm  \hspace*{.1in}  } |Q_r| + |K|
&>  (((1+\Delta)(1+\alpha)^3 - (1-\alpha)^3 + 1)\\
&\hspace*{.4in}-(\Delta(1+\alpha)^3 -  (1-\alpha)^3 + 1))\cdot N(\max\{0,t_r-2D\})\\
&=  (1 + \alpha)^3N(\max\{0,t_r-2D\})\\
&\ge |J|.
\end{align*}
This implies that $K$ and $Q_r$ intersect, since $K,Q_r \subseteq J$.
For each node $p$ in the intersection, $ts_p \ge \tau_w$
when $p$ sends its response to $r$ and, thus, $\tau_w\le \tau_r$. 


\end{proof}

\begin{lemma}
\label{lem:lin6}
For any two operations $op_1$ and $op_2$ in $\cal T$ such that $op_1$ finishes before $op_2$ starts, 
\\$(a)$  $ts(op_1) \leq ts(op_2)$ if $op_2$ is a read, and\label{lem:lin6a}
 \\$(b)$  $ts(op_1) < ts(op_2)$ if $op_2$ is a write.\label{lem:lin6b}
\end{lemma}
\begin{proof}
By definition of the timestamp of an operation, $ts(op_1) \le ts^{wp}(op_1)$ and $ts^{rp}(op_2) \le ts(op_2)$.
By Lemma~\ref{lem:lin5}, $ts^{wp}(op_1) \le ts^{rp}(op_2)$, and thus part $(a)$ follows.  Part $(b)$ follows from the observation that when $op_2$ is a write, $ts^{wp}(op_2) = ts^{rp}(op_2) + 1$.
\end{proof}

\begin{theorem}
 \label{thm:atomicity}
\AlgName{} ensures atomicity.
\end{theorem}

\begin{proof}
We show that, for every execution, there is a total order on the set of operations in $\cal T$ such that
 every read returns the value of the latest preceding write and, if an operation $op_1$ finishes before another operation $op_2$ begins, then $op_1$ is ordered before $op_2$.

Before describing the total order, we first argue that each write operation in $\cal T$ is assigned a unique timestamp.
  Recall that the timestamp of an operation $op$ executed by a node $p$ is the timestamp of $op$'s write phase, which is the ordered pair consisting of the values of $p$'s $seq$ and $id$ variables when $p$ executes Line~\ref{line:bcast update1} for $op$. 
   Note that the $id$ variable is equal to $p$ and the $seq$ variable is set to one greater than the largest sequence value observed during $op$'s read phase. These timestamps are unique because all nodes have unique ids, each node runs at most one write thread, and the writer remembers the last value that it assigned to $seq$. 

   {\em Claim:}  Consider any read $op_1$ in $\cal T$.  If the timestamp of $op_1$ is $(0,\bot)$, then $op_1$ returns $\bot$.  Otherwise there exists a write $op_2$ in $\cal T$ such that $ts(op_1) = ts(op_2)$ and the value returned by $op_1$ equals the value written by $op_2$.
\\This claim can be shown by a simple induction, based on the fact that every timestamp other than $(0,\bot)$ ultimately comes from Lines ~\ref{line:new timestamp}-\ref{line:new timestamp2} of Algorithm~\ref{algo:rw_client}.
  


The total order is constructed as follows.  First, order the write operations in order of their (unique) timestamps. Place each read with timestamp $(0,\bot)$ at the beginning of the total order. Place every other read immediately following the write operation it reads from.
 Finally, reorder all reads that are between the same two writes (or are before the first write or after the last write) according to the times when they start.
By the claim above, every read in the total order returns the value of the latest preceding write (or returns $\bot$ if there is no preceding write).




The rest of the proof shows that the total order respects the real-time order of non-overlapping operations in the execution.  Let $op_1$ and $op_2$ be two operations such that $op_1$ finishes before $op_2$ begins in the execution.  
 We do the proof in the following cases:
\begin{itemize}

\item \textbf{W-W:} If $op_1$ and $op_2$ are both writes, then Lemma~\ref{lem:lin6}$(b)$ implies that $ts(op_1) < ts(op_2)$ and thus the construction 
 orders $op_1$ before $op_2$.

\item \textbf{W-R:} Suppose $op_1$ is a write and $op_2$ is a read. By Lemma~\ref{lem:lin6}$(a)$ and the construction, $op_2$ is placed after the write $op_3$ that $op_2$ reads from. 
If $ts(op_1)  = ts(op_2)$ then $op_1 = op_3$ and $op_2$ is placed after $op_1$. If $ts(op_1) < ts(op_2)$ then $op_3$ is placed after $op_1$  as $ts(op_1) < ts(op_3)$ and thus $op_2$ is placed after $op_1$ in the total order.
 
\item \textbf{R-W:}  Suppose $op_1$ is a read and $op_2$ is a write.  By Lemma~\ref{lem:lin6}$(b)$, $ts(op_1) < ts(op_2)$.
Now, either $op_2$ is the first write in the execution and $op_1$'s timestamp is $(0,\bot)$ or there exists another write $op_3$ that $op_1$ reads from. 
   If $op_1$'s timestamp is $(0,\bot)$ then the construction orders $op_1$ before $op_2$. Otherwise, the construction orders $op_3$ before $op_2$.  Since $op_1$ is ordered after $op_3$ but before any subsequent write, $op_1$ precedes $op_2$ in the total order.
   
 
 


\item \textbf{R-R:} Finally, suppose that $op_1$ and $op_2$ are both reads. By Lemma~\ref{lem:lin6}$(a)$, $ts(op_1) \le ts(op_2)$. If $op_1$ and $op_2$ have the same timestamp, then they are placed after the same write (or before the first write) and the construction orders them based on their starting times. Since $op_1$ completes before $op_2$ starts, the construction places $op_1$ before $op_2$.  If $op_2$ has a timestamp greater than that of $op_1$, then $ts(op_2)$ cannot be $(0,\perp)$ and so, there is a write operation $op_3$ whose timestamp is greater than that of $op_1$ and equal to that of $op_2$. The construction places $op_1$ before $op_3$ and $op_2$ after $op_3$.

\end{itemize}

Thus CCREG ensures atomicity.


\end{proof}


The {\sc \AlgName{}} algorithm violates atomicity if our churn assumption $\alpha$ is violated.  This is demonstrated by the following
execution, in which large numbers of nodes enter and leave very quickly.

Let $|S_0| = n$ and let $p$ be some node in $S_0$.  Suppose the
following sequence of events occur before time $D$.  First, a set of
nodes, denoted $S_{new}$, enter the system, with $|S_{new}| = m \gg n$.
All join-related messages between $S_0 - \{p\}$ and $S_{new} \cup \{p\} $ take $D$ time, while the rest are very fast.  As a result, $p$
is the first joined node that nodes in $S_{new}$ hear from and they
use $p$'s estimate of the system size as being $n$ to calculate the
number of messages they should hear from before joining.  Thus all
nodes in $S_{new}$ join before time $D$ but no node in $S_0$ other
than $p$ knows about $S_{new}$ so far.

Second, immediately after joining, some node $q$ in $S_{new}$ invokes
{\em write}$(1)$.  All write-related messages between $S_0$ and
$S_{new}$ take $D$ time, while the rest are very fast.  $S_{new}$ is
sufficiently large that the write protocol completes for $q$ based
solely on hearing from nodes in $S_{new}$.  Thus the write completes
before time $D$ but no node in $S_0$ knows about the enters or the
write so far.

Third, immediately after the write finishes, all the nodes in
$S_{new}$ leave.  All leave-related messages between $S_0$ and
$S_{new}$ take $D$ time, while the rest are very fast.  Thus no node
in $S_0$ knows about the enters, the write, or the leaves so far.

Finally, immediately after the leaves, node $p' \ne p$ in $S_0$
invokes a read.  All read-related messages are very fast.  Node $p'$
uses its estimate of the system size as $n$ to decide how many
messages to wait for and is able to complete its 
read before time $D$
by hearing only from nodes in $S_0 - \{p\}$.  Since none of these
nodes knows anything about the write, the read returns 0, which
violates atomicity (as well as regularity).



\remove{
\begin{lemma}
\label{lem:gen1}
Suppose a node $p \not\in S_0$ receives an enter-echo message at time $t''$ from
a node $q$ that sent it at time $t'$ in response to an enter message from $p$.
If  $p$ is active at time $t+2D$ and
$q$ is active throughout
$[\max\{0,t'-2D\},t + D]$,
where $\max\{0,t''-2D\}\leq t \leq t_p^e$,
then \knows{p}{t+2D}{(\max\{0,t'-2D\},t]}.
\end{lemma}


\begin{proof}
Consider any node $r$ that enters, joins, or leaves at time $\hat{t}$, where $\max\{0,t'-2D\} < \hat{t} \leq t$.
If $q$ receives the message about this
change
from $r$ before the enter message from $p$, then
the change
is in $\mbox{\it Changes}_p^{t''} \subseteq \mbox{\it Changes}_p^{t+2D}$.
Otherwise,  $q$ receives the
message from $r$ after the enter message from $p$ and sends an echo message
in response by time $\hat{t}+D$.
Since $p$ receives this message from $q$ by time $\hat{t}+2D \leq t+2D$, it follows that
the change
is in $\mbox{\it Changes}_p^{t+2D}$.
Thus, \knows{p}{t+2D}{(\max\{0,t'-2D\},t]}.
\end{proof}
}





\remove{
it is sufficient for $|P_r|$ to be at least $f(\beta - 1/2)/(\beta c_1 - c_2)$,
where $c_1 = (1 - \alpha(\alpha + 2)(1 - \alpha)^2 + 1)/(1 - \alpha)^4$
and $c_2 = ((1 + \alpha)^3 + (1 + \alpha)^2)/(1 - \alpha)^2)/(1 - \alpha)^2$.

(*Presumably for this to be true, we need certain bounds on $\alpha$, $\beta$
and $N_{min}$?*)
} 

\remove{
\begin{lemma}
\label{lem:beforejoining}
For every node $p \not\in S_0$,\\
before $p$ joins, it receives an enter-echo message from a node that
is active throughout $[t_p^e-2D,t_p^e +D]$.
\end{lemma}
} 


\section{Related Work} \label{section:related}

A simple simulation of a single-writer, multi-reader register
in an asynchronous static network was presented by 
Attiya, Bar-Noy and Dolev~\cite{AttiyaBD1995}, henceforth called the ABD simulation.
Their paper also shows that it is impossible to simulate an atomic register in an asynchronous system if 
at least half
of the nodes in the system can be faulty.
It was followed by extensions that
reduce complexity~\cite{Attiya2000,DuttaGLC2004,GuerraouiL2004,GuerraouiM2007},
support multiple writers~\cite{LynchS1997},
or tolerate Byzantine failures~\cite{AbrahamCKM2006,AiyerAB2007,MartinAD2002,MalkhiR1998}.
To optimize load and resilience, the simple majority quorums
used in these papers can be replaced by other, more complicated,
quorum systems (e.g.,~\cite{MalkhiRWW2001,Vukolic2012}).

\remove{[[FE: Longer, more understandable descriptions of RAMBO and DynaStore are needed for the journal version. For example, we should explain what a system configuration is and what reconfiguration means. It would also be helpful to give some description of how they and [7] work.]]}

A survey of simulations of a atomic multi-writer, multi-reader 
register in a dynamic system with churn appears in  \cite{MusialNS2014CACM}.
The first such simulation was RAMBO~\cite{LynchS2002}.
Here, the notion of churn is abstracted by defining a sequence of \textit{quorum configurations}. Each quorum configuration consists of a set $S$ of nodes (which are called members) plus sets of read-quorums and write-quorums, each of which is a subset of $S$. 
The system supports \textit{reconfiguration}, 
in which an older quorum configuration is replaced by a newer one. 
RAMBO consists of three protocols:  \textit{ joiner}, \textit{ reader-writer}, and \textit{ recon}. The joiner protocol handles the joining of new nodes. The reader-writer protocol is responsible for executing the read and write operations using the read-quorums and write-quorums, as well as for performing garbage collection to remove old quorum configurations. The reads and writes are similar to the ABD simulation. 
 The recon protocol handles quorum configuration changes to install new quorum systems. Reconfiguration is done in two parts:  first, a member proposes a new quorum configuration. Second, these proposed configurations are reconciled by running an eventually-terminating distributed consensus algorithm (a version of the Paxos algorithm) among the members of the current quorum configuration. RAMBO requires intermittent periods of synchrony for the consensus to terminate. Reconfigurations can occur concurrently with  reads and writes. The model does not differentiate between nodes that crash and nodes that leave the system. 
 The algorithm guarantees atomicity of operations for all executions, even when there are arbitrary crashes (or leaves) and message loss.
  However, liveness is only ensured during periods when the system is sufficiently well-behaved with respect to synchrony, message loss, and churn.


DynaStore~\cite{AguileraKMS2011}  simulates an atomic  multi-writer, multi-reader  register in a dynamic system, without using consensus. 
The set of nodes that are in the system is called a {\em view}. 
The nodes start with some default initial \textit{view}.
The nodes in the current view can propose the addition
and removal of other nodes. 
\remove{\faith{For each view, there is a single-writer snapshot object that the nodes in that view use to announce proposed updates
to the current view.}
\faith{FAITH: WHY IS IT IMPORTANT TO MENTION THE SNAPSHOT OBJECT?}}
 The algorithm supports three types of operations: \textit{read}, \textit{write} and \textit{reconfig}.
Reads and writes are similar to the ABD simulation, with a read-phase
 followed by 
 a write-phase. 
The reconfig operation 
starts with a phase in which information about the new view is sent 
to a majority of nodes in the old view.
Then a read phase and a write phase are performed using the old view.
DynaStore ensures atomicity for all executions. To ensure liveness of operations, the algorithm makes two assumptions. First, at any point in time, the number of crashed nodes and the number of nodes  whose removal is pending (via reconfig) is a minority of the current view and of any pending future views. 
Second, it assumes that only a finite number of reconfiguration requests occur (i.e., churn is 
eventually 
 \textit{quiescent}).
Baldoni et al.~\cite{BaldoniBKR2009} proposed a model in which
the system size varies within a known range, that is, the nodes know
an upper and a lower bound on the system size. 
Their churn model has a similar flavor to ours, since churn never stops and
at most a \emph{constant} fraction of nodes enter and leave periodically. 
The authors implement a regular register in an eventually synchronous system. 
The algorithm has three protocols: 
\textit{ join\_register}, \textit{read}, and \textit{ write}. 
The join\_register module ensures that nodes join the system with sufficient knowledge about the system. 
The read and write protocols are similar to the ABD simulation.  
The simulation violates regularity if the churn assumption is violated. 
This can be shown with an argument similar to the one presented at 
the end of Section~\ref{section:proof}.


\remove{
in that at most a \emph{constant} fraction of nodes enter and leave periodically.
and there is an unknown upper bound on message delay.
However, in their model, the system size is assumed to be constant (and known to the nodes),
i.e., the number of nodes entering is the same as
the number of nodes leaving at each point in time.
Our model is more general, as we do not require that the system size is
always the same.
Instead, in our model, the system can grow, shrink,
or alternately grow and shrink.}

Baldoni et al. \cite{BaldoniBKR2009} also prove that it is impossible to
simulate a regular
register when there is no upper bound on message delay.
In this case, it does not help for nodes to announce when they are leaving,
since messages containing such announcements can be delayed for an arbitrarily long time.
Thus, {\em a node leaving is essentially the same as a crash.}
Their proof works by considering scenarios in which
at least half of the nodes fail. 
Then they invoke the lower bound in \cite{AttiyaBD1995},
which shows that simulating a register is impossible unless fewer than half the nodes are
faulty.
Their proof can be adapted to hold when there is an unknown upper bound, $D$,
on the message delay and half the nodes can be replaced during any time interval of length $D$,
provided that nodes are not required to announce when they leave.
Thus, in our model, 
there must be an upper bound on the fraction of nodes that can crash 
during any time interval of length $D$.
Also, it is necessary
that either nodes announce when they leave 
or there is an upper bound on the fraction of nodes 
that can leave during this time interval.

We summarize the  results of \cite{LynchS2002}, \cite{AguileraKMS2011} and \cite{BaldoniBKR2009}  and compare them with our algorithm  in Table~\ref{table:comparison}.
{In~\cite{SpiegelmanK2015} and~\cite{AguileraKMS2011}, it is claimed that termination of operations cannot be guaranteed unless the churn eventually stops. This claim does not contradict our result due to differences in the churn models.}
One of the contributions of this paper is to point out that by
making different, yet still reasonable, assumptions on churn it is possible
to get a solution with different, yet still reasonable, properties
and, in particular, to overcome the prior constraint that churn
must stop to ensure termination of operations.  That is, we are suggesting a
different point in the solution space.

\sapta{
\begin{table}[htbp!]
\resizebox{\textwidth}{!}{
    \begin{tabular}{|c|c|c|c|c|c|c|}
        \hline
       Algorithm  & Consistency &Synchrony & Consensus &Algorithm-Independent& Tolerates
       & Failure/Leave                 \\
       & Condition&Level & used  & Churn Assumption                        & Continuous Churn  
       &Model  \\\hline
         RAMBO~\cite{LynchS2002} &Atomic  & Intermittently    & Yes  &No&No, needs periods& 
         Leaves are  same                  \\
      &  &  Asynchronous&&& of quiescence &              as crashes \\
       & & and Synchronous  & &&&                       \\\hline
        DynaStore~\cite{AguileraKMS2011}&Atomic &  Asynchronous   & No&No & No, needs periods &
         Leaves are different     \\
    &        &      &&& of    quiescence   &     from crashes \\ \hline

        Baldoni      &Regular &  Eventually& No  &Yes      &     Yes    & 
                     Leaves are  same    \\
           et al.~\cite{BaldoniBKR2009} &  &  Synchronous                  &              &    &&  as crashes\\ \hline
\AlgName{}  &Atomic &  Asynchronous& No &Yes        & Yes    &
 Leaves are different            \\
 &  & & & &&    from crashes          \\
        \hline
    \end{tabular}
    }\caption{Comparison of our algorithm to RAMBO~\cite{LynchS2002}, DynaStore~\cite{AguileraKMS2011}, and Baldoni et al.~\cite{BaldoniBKR2009} }\label{table:comparison}
\end{table}
}

\section{Discussion}

We have shown how to simulate an atomic read/write register in a 
crash-prone asynchronous system where nodes can enter and leave.
The only assumptions are that the number of nodes entering and leaving
during each
time interval of length $D$
is at most a constant fraction of the current system size and the number of failures at any given time is a constant fraction of the system size.

It would be nice to improve the constants for
the churn rate and the maximum fraction of faulty nodes,
perhaps with a tighter analysis.Proving lower bounds or tradeoffs on these parameters is an interesting
avenue for future work.
In fact, it might be possible to completely avoid the bound $\alpha$
on the churn rate, by spreading out the handling of node joins and leaves:
To ensure a minimal number of nonfaulty nodes,
a node might need to obtain permission before leaving,
similarly to joins.
This may allow the algorithm to maintain safety even when
the churn is high.

Because of 
the bounded churn, it may be possible to implement a failure detector to get rid of crashed nodes, rather than relying on an external component for this.

\AlgName{} sends increasingly large Changes sets.
The amount of information communicated
might be reduced by sending only recent events,
or by removing very old events.
Another interesting research direction is to extend \AlgName{}
to tolerate more severe kinds of failures.

\bibliography{references}
\bibliographystyle{abbrv}

\end{document}